\documentclass[12pt, draftclsnofoot, onecolumn]{IEEEtran}
\usepackage[T1]{fontenc}
\usepackage{cite}
\usepackage{amsthm,amsmath}	 
\usepackage{graphicx}
\usepackage{epstopdf}
\usepackage{amssymb} 
\usepackage{algpseudocode,algorithm,algorithmicx}
\usepackage{subcaption}
\usepackage[font={small}]{caption}
\usepackage{breqn} 

\renewcommand{\vec}[1]{\mathbf{#1}}	

\newtheorem{theorem}{Theorem} 	
\newtheorem{proposition}{Proposition} 	
\newtheorem{definition}{Definition} 
 
\newtheorem{corollary}{Corollary}[theorem] 	%
\newtheorem{lemma}{Lemma}

\newcommand{\ignore}[1]{}
\algrenewcommand\algorithmicwhile{\textbf{nodes}}
\algrenewtext{EndWhile}{\algorithmicwhile\ \algorithmicend}

\algrenewcommand\algorithmicrequire{\textbf{Precondition:}}
\algrenewcommand\algorithmicensure{\textbf{Postcondition:}}

\DeclareMathOperator{\tr}{tr}

\ifCLASSINFOpdf
\else
\fi
\begin{document}
%
\title{Harnessing NLOS Components for Position and Orientation Estimation in 5G mmWave MIMO}
%
%
%

\author{Rico~Mendrzik,~\IEEEmembership{Student Member,~IEEE,}
        Henk~Wymeersch,~\IEEEmembership{Member,~IEEE},
				Gerhard~Bauch,~\IEEEmembership{Fellow,~IEEE},
				Zohair~Abu-Shaban,~\IEEEmembership{Member,~IEEE}%
\thanks{R. Mendrzik and  G. Bauch are with the Institute of Communications, Hamburg University of Technology, Hamburg,
 21073 Germany. H. Wymeersch is with the Department of Electrical Engineering, Chalmers University, Gothenburg,
 Sweden. Zohair Abu-Shaban is with the Research School of Engineering (RSEng) at the Australian National University (ANU). This work is supported, in part, by the Horizon2020 projects 5GCAR and HIGHTS (MG-3.5a-2014-636537), and the VINNOVA COPPLAR project, funded under Strategic Vehicle Research and Innovation Grant No. 2015-04849.}
}

\maketitle

\begin{abstract}
In the past, NLOS propagation was shown to be a source of distortion for radio-based positioning systems. Every NLOS component was perceived as a perturbation which resulted from the lack of temporal and spatial resolution of previous cellular systems. Even though 5G is not yet standardized, a strong proposal, which has the potential to overcome the problem of limited temporal and spatial resolution, is the massive MIMO millimeter wave technology. Based on this proposal, we reconsider the role of NLOS components for 5G position and orientation estimation purposes. Our analysis is based on the concept of Fisher information. We show that, for sufficiently high temporal and spatial resolution, NLOS components always provide position and orientation information which consequently increases position and orientation estimation accuracy. We show that the information gain of NLOS components depends on the actual location of the reflector or scatter. Our numerical examples suggest that NLOS components are most informative about the position and orientation of a mobile terminal when corresponding reflectors or scatterers are illuminated with narrow beams. 
\end{abstract}


%
\IEEEpeerreviewmaketitle

\section{Introduction}
\label{section_1}
%
%
%
%
\subsection{Motivation and State of the Art}
\IEEEPARstart{I}{n} many conventional wireless networks, multipath (MP) propagation is considered as a distorting effect, which cannot be leveraged for positioning of network nodes, when no prior information regarding the location of the corresponding point of incidence\footnote{In order to cover both reflectors and scatterers, we use the term \textit{point of incidence} in place of the location of a scatterer and the point of reflection of a reflector.} is available \cite{HSZWM2016,SW2010,SWW2010,SW2010b}. The reason is that the information enclosed in the waveform of the received signal is not rich enough to resolve the non-line-of-sight (NLOS) components in space and time. The fifth generation (5G) networks are expected to use signals in the millimeter wave (mmWave) band \cite{ABCHLSZ2014} and employ massive multiple input multiple output (MIMO) to compensate for the high path loss \cite{RRE2014,GTCRMVRMSN2014}. Particularly, mmWave MIMO systems operate at a carrier frequency  beyond 28 GHz using a large number of antennas at the base station and the mobile terminal \cite{ABCHLSZ2014,PK2011,RSMZAWWSSG2013,HGRRS2016,OER2015}. In the mmWave band, large contiguous frequency blocks are available enabling the support of high data rates \cite{GTCRMVRMSN2014,BH2015}. The large bandwidth in the mmWave band \cite{OER2015} result in high temporal resolution \cite{PAT2005}. Moreover, the short wavelength of mmWave signals makes it possible to accommodate a large number of antennas in a small area \cite{SGDSW2015,PK2011}. Hence large antenna arrays can be expected for both base stations as well as mobile terminals. Large antenna arrays, in turn, allow for extremely narrow beams which enable accurate spatial resolution in the angular domain \cite{LCP2005,LSS2009}. Even though the positioning capabilities of mmWave MIMO in 5G are not yet fully explored, the high temporal and spatial resolutions of mmWave MIMO systems suggest that NLOS components can be resolved and hence can be harnessed for position and orientation estimation.

The fundamental limits of position and orientation estimation using mmWave MIMO in 5G have been recently investigated in \cite{GGD2017,SGDSW2017,SZASW2017}. In \cite{GGD2017}, a single anchor localization scheme is presented for indoor scenarios. The Fisher information matrix (FIM) of the position and orientation parameters as well as the NLOS parameters was presented. Based on this FIM the position error bound (PEB) and orientation error bound (OEB) were derived numerically. Different array configurations were considered. It was shown that increasing the number of antenna array elements increases the localization accuracy. In \cite{SGDSW2017}, the FIM of all channel parameters was presented. Using the geometric relationship of the channel parameters and the position and orientation-related parameters, the FIM of the position and orientation-related parameters was derived in closed-form. Moreover, the PEB and the OEB were determined numerically, and algorithms which attain the previously determined bounds were also presented. It was shown numerically that even in NLOS situations, positioning with reasonable accuracy is possible. In \cite{SZASW2017},  fundamental limits of position and orientation estimation for uplink and downlink in 3D-space were presented. The FIM of the channel parameters was derived in a closed form similar to \cite{SGDSW2017}, which provided the FIM of the 2D channel parameters. Moreover, the structure of this FIM was analyzed and it was shown to become block diagonal when the bandwidth and the number of receive and transmit antennas are sufficiently large. In contrast to \cite{SGDSW2017}, which considered uniform linear arrays, \cite{SZASW2017} presented the derivation of the PEB and the OEB in closed-form for any arbitrary antenna array structure. The PEB and the OEB were derived similarly to \cite{SGDSW2017}. In addition, the influence of different array types on the PEB and the OEB was investigated. Moreover, differences in the uplink and downlink were considered. 

NLOS components have already been proven to be useful for indoor navigation \cite{LMMWH2016,KLMHW2016,WMLSGTHDMCW2016,GJWZDF2016,GMMUJD2016}. In \cite{GJWZDF2016,GMMUJD2016}, a two-stage approach is adopted to estimate and track the position of the mobile terminal and the positions of virtual anchors. Virtual anchors mimic a line-of-sight (LOS) transmission for every NLOS component. In the first stage of the approach in \cite{GJWZDF2016,GMMUJD2016}, the complex channel gains, and delays are estimated and tracked. Based on these results, the position of the mobile terminal and the locations of the virtual anchors are estimated and tracked. The approaches in \cite{LMMWH2016,KLMHW2016,WMLSGTHDMCW2016} leverage the huge bandwidth of ultra-wideband (UWB) signals in order to resolve NLOS components in time. Each NLOS component is then associated with a virtual anchor. NLOS components can be associated with virtual anchors using, e.g., belief propagation \cite{LMMWH2016} or optimal sub-pattern assignment \cite{KLMHW2016}. In order to reliably associate NLOS components with virtual anchors, multiple observations and mobility of the mobile terminal are required. Virtual anchors and the unknown position of the node are tracked over time using different filters, e.g., belief propagation \cite{LMMWH2016} or the extended Kalman filter \cite{KLMHW2016}. 
The key difference of mmWave MIMO schemes in comparison with the works in \cite{LMMWH2016,KLMHW2016,WMLSGTHDMCW2016,GJWZDF2016,GMMUJD2016} is that they do not rely on the mobility of the mobile terminal to harness information from NLOS components. A snapshot (one transmission burst from the base station) is sufficient to exploit the information which NLOS components provide.

\subsection{Contribution and Paper Organization}
In \cite{SZASW2017,SGDSW2017,GGD2017}, it was numerically shown that position and orientation estimation accuracy can benefit from NLOS components. However, the influence of the location of the base station, mobile terminal, and points of incidence of NLOS components is not well understood. The convoluted structure of the FIM of the channel parameters makes the analysis of the impact of NLOS components complicated. In our work, we build upon \cite{SZASW2017} and employ a simplified FIM of the channel parameters. Using a geometric transformation like in \cite{SZASW2017,SGDSW2017}, we obtain a simplified FIM in the position, orientation, and points of incidence domain. In order to study the impact of NLOS components on the position and orientation estimation accuracy, we employ the notion of the equivalent FIM (EFIM) \cite{SW2007}. Firstly, we determine the EFIM of the position and orientation. Then, we decompose this EFIM in order to analyze and reveal the effect of NLOS components. Our contributions are summarized as follows:
\begin{itemize}
	\item  Assuming a large number of receive and transmit antennas as well as a large bandwidth, we derive an expression for the EFIM of the position and orientation, and we show that this EFIM can be written as the sum of rank one matrices, where each NLOS component contributes a distinct rank one matrix. 
	\item We show that each NLOS component contributes position and orientation information to the EFIM which reduces the PEB and OEB. We show that NLOS components provide significant position and orientation information if and only if angle-of-arrival (AOA), angle-of-departure (AOD), and time-of-arrival (TOA) can be estimated accurately.	
	\item We derive the amount and direction of information in a closed form showing its relation to the geometry.
\end{itemize}
The rest of the paper is organized as follows. Section \ref{sec:sys_model} discusses our system model, and section \ref{sec:FIM} reviews the simplified FIM of the channel parameters from \cite{SZASW2017}. Our main results are presented in section \ref{sec:EFIM}, where we derive the EFIM of the position and orientation, decomposition the EFIM, and show the information gain of NLOS components. Section \ref{sec:numerical_example} contains numerical examples. The paper is concluded in section \ref{sec:conclusion}.

\textit{Notation}: Throughout this paper, we will stick to the following notational conventions. Scalars are denoted in italic, e.g. $x$.	Lower case boldface indicates a column vector, e.g. $\vec{x}$, while upper case boldface denotes a matrix, e.g. $\vec{X}$. Matrix elements are denoted by $[\vec{X}]_{i,j}$ where $i$ refers to rows and $j$ refers to columns, while $[\vec{X}]_{i:l,j:k}$ selects the sub-matrix of $\vec{X}$ between the rows $i$ to $l$ and the columns $j$ to $k$.  Matrix transpose is indicated by superscript $\mathrm{T}$, e.g. $\vec{X}^{\mathrm{T}}$, while the superscript $\mathrm{H}$ refers to the transpose conjugate complex. Matrix trace is expressed by $\tr(\vec{X})$ and matrix determinant is indicated as $|\vec{X}|$. The Euclidean norm is denoted by $\left\|\cdot\right\|$, e.g. $\left\|\vec{x}\right\|$.

\section{System Model}
\label{sec:sys_model}
In this section, we first describe the geometry of the considered problem. Secondly, we specify the transmitter and the channel models. We conclude the section with the model of the receiver.

\subsection{Geometry}
\begin{figure}[t]%
\centering
\includegraphics[width=0.6\columnwidth]{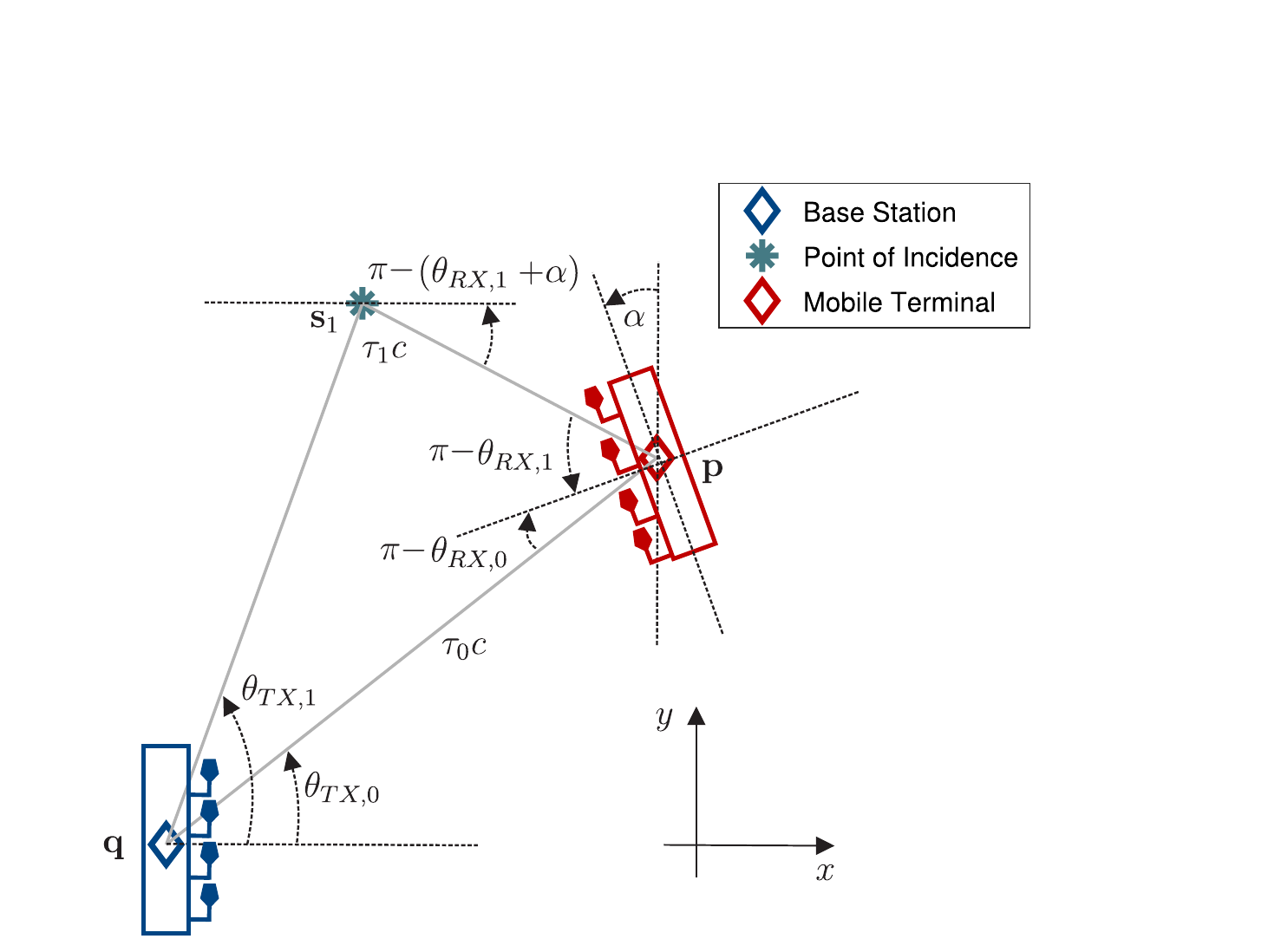}%
\caption{\textit{Geometry of the scenario -} A mobile terminal with unknown position and orientation attempts to localize itself and determine its orientation using the signal received from a base station. The base station has known location and orientation. Single-bounce NLOS paths and a direct path are considered.}%
\label{fig:system_geometry}%
\end{figure}
We consider a mobile terminal which aims to estimate its own location and orientation in 2D space, based on the downlink signal received from the base station. The position and orientation of the base station are perfectly known to the mobile terminal. We assume that mobile terminal and base station are synchronized\footnote{The synchronization assumption can be removed by considering a two-way protocol \cite{LS2002,ZLK2007}.}. An illustration of the scenario is depicted in Fig. \ref{fig:system_geometry}. The base station and mobile terminal are equipped with an array of $N_{\mathrm{TX}}$ transmit antennas and $N_{\mathrm{RX}}$ receive antennas, respectively. The array of the base station has arbitrary but known geometry. The orientation of the base station array is denoted by $\phi$. The centroid of the base station array is located at the position $\vec{q}=[q_{\mathrm{x}}, q_{\mathrm{y}}]^{\mathrm{T}}$. The centroid of the array of the mobile terminal\footnote{From now onwards, we will treat the centroid of the array of the base station and mobile terminal as the position of the base station and mobile terminal, respectively.} is located at $\vec{p}=[p_{\mathrm{x}}, p_{\mathrm{y}}]^{\mathrm{T}}$. We assume that its array geometry is known while the orientation of the array $\alpha$ is unknown. 

\subsection{Transmitter Model}
We consider mmWave in combination with massive MIMO. In particular, the transmitter transmits $\tilde{\vec{s}}(t)\triangleq\sqrt{E_{\mathrm{s}}}\vec{F}\vec{s}(t)$, where $E_{\mathrm{s}}$ denotes the energy per symbol, $\vec{F}\triangleq[\vec{f}_1,\vec{f}_2,...,\vec{f}_{{N}_{\mathrm{B}}}]$ is a precoding matrix with $N_{\mathrm{B}}$ simultaneously transmitted beams, and $\vec{s}(t)\triangleq [s_1(t),...,s_{{N}_{\mathrm{B}}}(t)]^{\mathrm{T}}$ is the vector of pilot signals. The pilot signal of the $l^{\text{th}}$ beam is given by
\begin{equation}
s_{l}(t)\triangleq\sum_{m=0}^{N_{s}-1} d_{l,m}p(t-mT_{\mathrm{s}}),
\label{eq:pilot_symbols}
\end{equation}
where $N_s$ denotes the number of pilot symbols per beam, $T_s$ is the symbol duration, $d_{l,m}$ are independent and identically distributed (IID) unit energy pilot symbols with zero mean which are transmitted over the $l^{\text{th}}$ beam with the unit-energy pulse $p(t)$. 
The $l^{\text{th}}$ column of $\vec{F}$ contains a directional beam pointing towards the azimuth angle $\theta_{\mathrm{BF},l}$ 
\begin{equation}
\vec{f}_l(\theta_{\mathrm{BF},l})\triangleq\frac{1}{\sqrt{N_{\mathrm{B}}}}\vec{a}_{\mathrm{TX},l}(\theta_{\mathrm{BF},l}),
\label{eq:individual_beam}
\end{equation}
where $\vec{a}_{\mathrm{TX},l}$ is the unit-norm array response vector given by \cite{VT2002}
\begin{equation}
\vec{a}_{\mathrm{TX},l}(\theta_{\mathrm{TX},l})\triangleq \frac{1}{\sqrt{N_{\mathrm{TX}}}}\exp(-j\boldsymbol\Delta_{\mathrm{TX}}^{\mathrm{T}}\vec{k}(\theta_{\mathrm{TX},l})), 
\label{eq:antenna_TX_reponse}
\end{equation}
where $\vec{k}(\theta_{\mathrm{TX},l})=\frac{2\pi}{\lambda}[\cos(\theta_{\mathrm{TX},l}),\sin(\theta_{\mathrm{TX},l})]^{\mathrm{T}}$ is the wavenumber vector, $\lambda$ is the wavelength, $\boldsymbol\Delta_{\mathrm{TX}}\triangleq[\vec{u}_{\mathrm{TX},1},\vec{u}_{\mathrm{TX},2},...,\vec{u}_{\mathrm{TX},N_{\mathrm{TX}}}]$ is a $2\times N_{\mathrm{TX}}$ matrix which contains the positions of the transmit antenna elements in 2D Cartesian coordinates in its columns, i.e. the $n^{\text{th}}$ column of $\boldsymbol\Delta_{\mathrm{TX}}$ is given by $\vec{u}_{\mathrm{TX},n}\triangleq [x_{\mathrm{TX},n},y_{\mathrm{TX},n}]^{\mathrm{T}}$.
To normalize the transmitted power, we set $\tr\left(\vec{F}^{\mathrm{H}}\vec{F}\right)=1$ and $\mathbb{E}\left\{\vec{s}(t)\vec{s}(t)^{\mathrm{H}}\right\}=\vec{I}_{N_{\mathrm{B}}}$, where $\vec{I}_{N_{\mathrm{B}}}$ is the $N_{\mathrm{B}}$-dimensional identity matrix.
\subsection{Channel Model}
We assume $K\geq 1$ distinct paths between the base station and the mobile terminal. Using mmWave massive MIMO, the number of paths is small \cite{PK2011}. The line-of-sight (LOS) path - if it exists - is denoted by $k=0$, while $k>0$ correspond to NLOS components. Due to the high path loss and the high directionality of the transmitted beams, NLOS components are assumed to originate from single-bounce scattering\footnote{Scatterers are objects that are much smaller than the wavelength of the signal.} or reflection\footnote{Reflectors are objects with a specific reflection point that are much larger than the wavelength of the signal.} only \cite{SGDSW2017,HSZWM2016,GGD2017,DS2014}. We denote the reflecting point and the location of the scatterer by the \textit{point of incidence} $\vec{s}_k=[s_{\mathrm{x},k},s_{\mathrm{y},k}]^{\mathrm{T}}$. Considering Fig. \ref{fig:system_geometry} it can be seen that each path is associated with three distinct channel parameters, namely AOA, AOD, and TOA, where AOA, AOD, and TOA of the $k^{\text{th}}$ path are denoted by $\theta_{\mathrm{RX},k}$, $\theta_{\mathrm{TX},k}$, and $\tau_k$, respectively. Assuming a narrow-band array model\footnote{We assume that $A_{\text{max}}<<c/B$, where $A_{\text{max}}$ is the maximum array aperture size, $c$ is the speed of light, and $B$ is the system bandwidth.}, the channel impulse response is given by
\begin{equation}
\vec{H}(t)=\sum_{k=0}^{K-1}\underbrace{\sqrt{N_{\mathrm{RX}}N_{\mathrm{TX}}}h_k\vec{a}_{\mathrm{RX},k}(\theta_{\mathrm{RX},k})\vec{a}_{\mathrm{TX},k}^{H}(\theta_{\mathrm{TX},k})}_{\vec{H}_k} \delta\left(t-\tau_k\right),
\label{eq:channel}
\end{equation}
where $h_{k}=h_{\mathrm{R},k}+j h_{\mathrm{I},k}$ is the complex path gain while $\vec{a}_{\mathrm{TX},k}(\theta_{\mathrm{TX},k})$ and $\vec{a}_{\mathrm{RX},k}(\theta_{\mathrm{RX},k})$ denote the unit-norm array response vectors of the $k^{\text{th}}$ path at the transmitter and receiver, respectively. Note that $\vec{a}_{\mathrm{TX},k}(\theta_{\mathrm{TX},k})$ is explicitly defined in \eqref{eq:antenna_TX_reponse}, while $\vec{a}_{\mathrm{RX},k}(\theta_{\mathrm{RX},k})$ can be defined analogously by \eqref{eq:antenna_TX_reponse} with matching subscripts.
\subsection{Receiver Model}
The noisy observed signal at the receiver is given by
\begin{equation}
\vec{r}(t)\triangleq \sum_{k=0}^{K-1} \sqrt{E_{\mathrm{s}}}\vec{H}_k \vec{F} \vec{s}(t-\tau_k)+\vec{n}(t), \quad t\in[0,N_sT_s],
\label{eq:receive_signal}
\end{equation}
where $\vec{n}(t)=[n_1(t),n(2),...,n_{N_{RX}}(t)]^{\mathrm{T}}$ is zero-mean additive white Gaussian noise (AWGN) with PSD $N_0$. Similar to \cite{FNBAS2010,VV2010}, we assume that a low-noise amplifier and a passband filter is attached to every receive antenna. This assumption might seem restrictive for the practical application, yet it simplifies the analysis of the EFIM. It can be regarded as the receiver architecture which results in the lowest PEB and OEB. 
\section{Fisher Information Matrix of the Channel Parameters}
\label{sec:FIM}
In this section, we first define the estimation problem and state the FIM of the channel parameters. We conclude the section with a brief summary of the results of \cite{SZASW2017}, which allow for a simplification of the FIM of the channel parameters.
\subsection{Definition}
We first we define the vector of channel parameters 
\begin{equation}
\boldsymbol\eta \triangleq [\boldsymbol\theta_{\mathrm{RX}}^{\mathrm{T}},\boldsymbol\theta_{\mathrm{TX}}^{\mathrm{T}},\boldsymbol\tau^{\mathrm{T}},\vec{h}_{\mathrm{R}}^{\mathrm{T}},\vec{h}_{\mathrm{I}}^{\mathrm{T}}]^{\mathrm{T}},
\label{eq:eta_channel}
\end{equation}
where we collect the AOAs, AODs, TOAs, and channel gains in the vectors 
\newline $\boldsymbol\theta_{\mathrm{RX}}\triangleq [\theta_{\mathrm{RX},0},\theta_{\mathrm{RX},1},...,\theta_{\mathrm{RX},K-1}]^{\mathrm{T}}$, $\boldsymbol\theta_{\mathrm{TX}}\triangleq [\theta_{\mathrm{TX},0},\theta_{\mathrm{TX},1},...,\theta_{\mathrm{TX},K-1}]^{\mathrm{T}}$, $\boldsymbol\tau\triangleq [\tau_0, \tau_1,...,\tau_{K-1}]^{\mathrm{T}}$, $\vec{h_{\mathrm{R}}}\triangleq [h_{\mathrm{R,0}},h_{\mathrm{R,1}},...,h_{\mathrm{R},K-1}]^{\mathrm{T}}$, and $\vec{h_{\mathrm{I}}}\triangleq [h_{\mathrm{I,0}},h_{\mathrm{I,1}},...,h_{\mathrm{I},K-1}]^{\mathrm{T}}$, respectively.	The corresponding FIM is given by
\begin{equation}
\vec{J}_{\boldsymbol\eta}\triangleq
\begin{bmatrix}
\vec{J}_{\boldsymbol\theta_{\mathrm{RX}}\boldsymbol\theta_{\mathrm{RX}}} & \vec{J}_{\boldsymbol\theta_{\mathrm{RX}}\boldsymbol\theta_{\mathrm{TX}}} & \cdots& \vec{J}_{\boldsymbol\theta_{\mathrm{RX}}\vec{h}_{\mathrm{I}}}\\
\vec{J}_{\boldsymbol\theta_{\mathrm{RX}}\boldsymbol\theta_{\mathrm{TX}}}^{\mathrm{T}} & \ddots & \cdots & \vdots \\
\vdots & \cdots &\ddots & \vdots \\
\vec{J}_{\boldsymbol\theta_{\mathrm{RX}}\vec{h}_{\mathrm{I}}}^{\mathrm{T}} & \cdots & \cdots & \vec{J}_{\vec{h}_{\mathrm{I}}\vec{h}_{\mathrm{I}}}\\
\end{bmatrix},
\label{eq:FIM_channel}
\end{equation}
where each entry of the FIM of the channel parameters can be computed according to\footnote{This result holds whenever the signal is observed under additive white Gaussian noise (AWGN).} \cite{K1993}
\begin{equation}
\left[\vec{J}_{\boldsymbol\eta}\right]_{u,v} \triangleq \frac{1}{N_0} \int_{0}^{N_{\mathrm{s}}T_{\mathrm{s}}}\mathbb{E}_{\mathrm{a}}\left[ \mathfrak{R}\left\{ \frac{\partial \boldsymbol\mu_{\boldsymbol\eta}^{\mathrm{H}}(t)}{\partial[\boldsymbol\eta]_u}\frac{\partial \boldsymbol\mu_{\boldsymbol\eta}(t)}{\partial[\boldsymbol\eta]_v}  \right\}\right] \text{d}t.
\label{eq:FIM_u_v_entry}
\end{equation}
In \eqref{eq:FIM_u_v_entry}, $\mathbb{E}_{\mathrm{a}}\left[\cdot\right]$ denotes the expectation with respect to the pilot symbols, $\mathfrak{R}\left\{\cdot\right\}$ is the real part of the argument, and $\boldsymbol\mu_{\boldsymbol\eta}(t)$ is defined as the noise-free observation
\begin{equation}
\boldsymbol\mu_{\boldsymbol\eta}(t)=\sum_{k=0}^{K-1} \sqrt{E_s}\vec{H}_k \vec{F} \vec{s}(t-\tau_k).
\label{eq:mu_vector}
\end{equation}
The FIM is related to the estimation error covariance matrix of any unbiased estimator via the information inequality \cite{A1993,TB2007,TBT2013}
\begin{equation}
	\mathbb{E}_a\left[ (\boldsymbol\eta-\hat{\boldsymbol\eta})(\boldsymbol\eta-\hat{\boldsymbol\eta})^{\mathrm{T}} \right] \succeq \vec{J}_{\boldsymbol\eta}^{-1},
\label{eq:information_inequality}
\end{equation}
where $\hat{\boldsymbol\eta}$ is the estimate of $\boldsymbol\eta$ and $\vec{A}\succeq \vec{B}$ is equivalent to $\vec{A}-\vec{B}$ being positive semi-definite. The inequality in \eqref{eq:information_inequality} is the well-known Cram\'er-Rao lower bound (CRLB).

\subsection{Simplification}
The blocks of the FIM in \eqref{eq:FIM_channel} obey certain scaling laws when the number of receive and transmit antennas, as well as the bandwidth become sufficiently large\footnote{It was shown in \cite{SZASW2017} that the approximation error of the PEB due to the simplification of the FIM is fairly small even under realistic assumption on the bandwidth ($B=125$ MHz) and the 2D array sizes ($N_{\mathrm{TX/RX}}=12\times 12$).}. In particular, it was shown in \cite[section III-B]{SZASW2017}  that some blocks can be well approximated by diagonal matrices, while others become zero matrices.
In the following, we provide a brief summary of the results from \cite{SZASW2017}. For more details, the reader is directly referred to \cite{SZASW2017}. 

Let $\vec{I}_{\mathrm{K}}$ and $\vec{0}_{\mathrm{K}}$ be the $K\times K$ identity and all-zeros matrix, respectively. We denote the Hadamard product by $\odot$, and make the following remarks:
\begin{enumerate}
	\item Since the AOAs of the different paths are assumed to be distinct, the steering vectors at the receiver do not interact considerably with each other, i.e. $\left\|\vec{a}^{\mathrm{H}}_{\mathrm{RX},u}\vec{a}_{\mathrm{RX},v}\right\|\ll \left\|\vec{a}^{\mathrm{H}}_{\mathrm{RX},u}\vec{a}^{\mathrm{H}}_{\mathrm{RX},u}\right\|,~ u\neq v$. Hence AOAs can be estimated independently and $\vec{J}_{\boldsymbol\theta_{\mathrm{\mathrm{RX}}}\boldsymbol\theta_{\mathrm{\mathrm{RX}}}}$ becomes diagonal, i.e. $\tilde{\vec{J}}_{\boldsymbol\theta_{\mathrm{\mathrm{RX}}}\boldsymbol\theta_{\mathrm{\mathrm{RX}}}}\approx \vec{I}_{\mathrm{K}} \odot \vec{J}_{\boldsymbol\theta_{\mathrm{\mathrm{RX}}}\boldsymbol\theta_{\mathrm{\mathrm{RX}}}}.$
	\item  The spatial cross-correlation of the transmitted beams decreases when the number of transmit antennas increases because the beams become narrower. Hence AODs can be estimated independently and $\vec{J}_{\boldsymbol\theta_{\mathrm{TX}}\boldsymbol\theta_{\mathrm{TX}}}$ becomes diagonal, i.e. $\tilde{\vec{J}}_{\boldsymbol\theta_{\mathrm{TX}}\boldsymbol\theta_{\mathrm{TX}}}\approx \vec{I}_{\mathrm{K}} \odot \vec{J}_{\boldsymbol\theta_{\mathrm{TX}}\boldsymbol\theta_{\mathrm{TX}}}.$
	\item  The NLOS cross-correlation vanishes as the bandwidth of the signal becomes large
since the paths can be resolved independently in time. Hence TOAs can be estimated independently and $\vec{J}_{\boldsymbol\tau\boldsymbol\tau}$ becomes diagonal, i.e. $\tilde{\vec{J}}_{\boldsymbol\tau\boldsymbol\tau}\approx \vec{I}_{\mathrm{K}} \odot \vec{J}_{\boldsymbol\tau\boldsymbol\tau}.$
	\item As a consequence of the previous results, the channel gains can be estimated independently and $\vec{J}_{\vec{h}_{\mathrm{R}}\vec{h}_{\mathrm{R}}}$, as well as $\vec{J}_{\vec{h}_{\mathrm{I}}\vec{h}_{\mathrm{I}}}$ become diagonal, i.e. $\tilde{\vec{J}}_{\vec{h}_{\mathrm{R}}\vec{h}_{\mathrm{R}}}\approx \vec{I}_{\mathrm{K}} \odot \vec{J}_{\vec{h}_{\mathrm{R}}\vec{h}_{\mathrm{R}}}$ and 
	\newline $\tilde{\vec{J}}_{\vec{h}_{\mathrm{I}}\vec{h}_{\mathrm{I}}}\approx \vec{I}_{\mathrm{K}} \odot \vec{J}_{\vec{h}_{\mathrm{I}}\vec{h}_{\mathrm{I}}}$.
	\item All off-diagonal blocks, except for $\vec{J}_{\boldsymbol\theta_{\mathrm{TX}}\vec{h}_{\mathrm{R}}}$ and $\vec{J}_{\boldsymbol\theta_{\mathrm{TX}}\vec{h}_{\mathrm{I}}}$, in \eqref{eq:FIM_channel} become zero. It was shown in \cite{SZASW2017} that the real and imaginary part of the $k^{\mathrm{th}}$ channel gain couple only with the AOD of the $k^{\mathrm{th}}$ path, i.e.  $\tilde{\vec{J}}_{\boldsymbol\theta_{\mathrm{TX}}\vec{h}_{\mathrm{R}}}\approx \vec{I}_{\mathrm{K}} \odot \vec{J}_{\boldsymbol\theta_{\mathrm{TX}}\vec{h}_{\mathrm{R}}}$ and $\tilde{\vec{J}}_{\boldsymbol\theta_{\mathrm{TX}}\vec{h}_{\mathrm{I}}}\approx \vec{I}_{\mathrm{K}} \odot \vec{J}_{\boldsymbol\theta_{\mathrm{TX}}\vec{h}_{\mathrm{I}}}.$ 
\end{enumerate}
Thus, when the bandwidth of the signal is large and number of receive and transmit antennas is also large, $\vec{J}_{\boldsymbol\eta}$ can be well approximated by 
\begin{equation}
\tilde{\vec{J}}_{\boldsymbol\eta}
\triangleq
\begin{bmatrix}
\tilde{\vec{J}}_{\boldsymbol\theta_{\mathrm{RX}}\boldsymbol\theta_{\mathrm{RX}}} & \vec{0}_{\mathrm{K}} & \vec{0}_{\mathrm{K}}&  \vec{0}_{\mathrm{K}} & \vec{0}_{\mathrm{K}}\\
\vec{0}_{\mathrm{K}} & \tilde{\vec{J}}_{\boldsymbol\theta_{\mathrm{TX}}\boldsymbol\theta_{\mathrm{TX}}} & \vec{0}_{\mathrm{K}} & \tilde{\vec{J}}_{\boldsymbol\theta_{\mathrm{TX}}\vec{h}_{\mathrm{R}}} & \tilde{\vec{J}}_{\boldsymbol\theta_{\mathrm{TX}}\vec{h}_{\mathrm{I}}} \\
\vec{0}_{\mathrm{K}} & \vec{0}_{\mathrm{K}} & \tilde{\vec{J}}_{\boldsymbol\tau\boldsymbol\tau} & \vec{0}_{\mathrm{K}} &\vec{0}_{\mathrm{K}} \\
\vec{0}_{\mathrm{K}} & \tilde{\vec{J}}_{\boldsymbol\theta_{\mathrm{TX}}\vec{h}_{\mathrm{R}}}^{\mathrm{T}} & \vec{0}_{\mathrm{K}}  & \tilde{\vec{J}}_{\vec{h}_{\mathrm{R}}\vec{h}_{\mathrm{R}}} &\vec{0}_{\mathrm{K}} \\
\vec{0}_{\mathrm{K}} & \tilde{\vec{J}}_{\boldsymbol\theta_{\mathrm{TX}}\vec{h}_{\mathrm{I}}}^{\mathrm{T}} & \vec{0}_{\mathrm{K}} & \vec{0}_{\mathrm{K}} &\tilde{\vec{J}}_{\vec{h}_{\mathrm{I}}\vec{h}_{\mathrm{I}}} \\
\end{bmatrix}.
\label{eq:simplfied_FIM}
\end{equation}

\section{Fisher Information Matrix of the Position-related Parameters}
\label{sec:EFIM}
Motivated by the findings of the previous subsection, we first reorder the parameters of the simplified FIM $\tilde{\vec{J}}_{\boldsymbol\eta}$ in \eqref{eq:simplfied_FIM}. Subsequently, we transform the resulting FIM to the position, orientation, and point of incidence domain. Then, we determine the EFIM of the position and orientation, which we decompose to analyze the impact of NLOS paths.

For mathematical convenience, we reorder the parameter vector $\boldsymbol\eta$ as follows 
\begin{equation}
\breve{\boldsymbol\eta}\triangleq \left[\tau_0,\theta_{\mathrm{TX},0},h_{\mathrm{R,0}},h_{\mathrm{I,0}},\theta_{\mathrm{RX},0},...,\tau_{K-1},\theta_{\mathrm{TX},K-1},h_{\mathrm{R},K-1},h_{\mathrm{I},K-1},\theta_{\mathrm{RX},K-1}\right]^{\mathrm{T}}.
\label{eq:eta_tilde}
\end{equation}
Reordering the parameter vector of the FIM results in a permutation of the entries of the FIM in \eqref{eq:simplfied_FIM}. The reordered FIM is given by 
\begin{equation}
\vec{J}_{\breve{\boldsymbol\eta}} 
\triangleq
\vec{P}_{\pi}\tilde{\vec{J}}_{\boldsymbol\eta},
\label{eq:J_eta_bar}
\end{equation}
where $\vec{P}_{\pi}$ is a permutation matrix of size $5K \times 5K$ which is given by
\begin{equation}
\vec{P}_{\pi} \triangleq
[\vec{e}_{2K+1},\vec{e}_{K+1},\vec{e}_{3K+1},\vec{e}_{4K+1},\vec{e}_{1},\cdots,\vec{e}_{3K},\vec{e}_{2K},\vec{e}_{4K},\vec{e}_{5K},\vec{e}_{K}]^{\mathrm{T}},
\label{eq:permutation}
\end{equation}
where $\vec{e}_{k}$ denotes the $k^{\text{th}}$ 	unit vector of the standard basis. Hence the FIM of the channel parameters has the following structure
\begin{equation}
{\vec{J}}_{\breve{\boldsymbol\eta}}=\mathrm{blkdiag}\left({\vec{J}}_{\breve{\boldsymbol\eta}_1},...,{\vec{J}}_{\breve{\boldsymbol\eta}_{K-1}}\right).
\label{eq:}
\end{equation}
The FIM block of the $k^{\mathrm{th}}$ path ${\vec{J}}_{\breve{\boldsymbol\eta}_k}$ is given by
\begin{equation}
{\vec{J}}_{\breve{\boldsymbol\eta}_k}
\triangleq
\left[
\begin{array}{c@{}c@{}c}
1/\sigma^2_{\tau_k} & \mathbf{0} & \mathbf{0} \\
  \mathbf{0} & \underbrace{\left[\begin{array}{ccc}
                       1/\tilde{\sigma}^2_{\theta_{\mathrm{TX},k}} & b_{\mathrm{R},k} & b_{\mathrm{I},k}\\ 
                       b_{\mathrm{R},k} & 1/\sigma^2_{\mathrm{h}_{\mathrm{R},k}} & 0\\
                       b_{\mathrm{I},k} & 0 & 1/\sigma^2_{\mathrm{h}_{\mathrm{I},k}}\\
                      \end{array}\right]}_{\vec{J}_{\theta_{\mathrm{TX},k}\vec{h}_k}} & \mathbf{0}\\
\mathbf{0} & \mathbf{0} & 1/\sigma^2_{\theta_{\mathrm{RX},k}} \\
\end{array}\right],
\label{eq:J_k}
\end{equation}
where we abbreviate the diagonal entries of the reordered FIM matrix of the channel parameters by the respective $1/\sigma^2$-terms, while $b_{\mathrm{R},k}$ and $b_{\mathrm{I},k}$ denote the corresponding entries of $\tilde{\vec{J}}_{\boldsymbol\theta_{\mathrm{TX}}\vec{h}_{\mathrm{R}}}$ and $\tilde{\vec{J}}_{\boldsymbol\theta_{\mathrm{TX}}\vec{h}_{\mathrm{I}}}$, respectively. Since, for position and orientation estimation, we are mainly interested in the pairs of AOA, AOD, and TOA of every path, we combine the uncertainty of $h_{\mathrm{R,k}}$ and $h_{\mathrm{I,k}}$ in the AOD-related term. For that, we will use the notion of the equivalent FIM (EFIM) from \cite{SW2007}, to consider only the information concerning AOA, AOD, and TOA. The EFIM is a measure of the information corresponding to certain parameters, while accounting for uncertainties of other (unknown) parameters. 
\begin{definition}
Given a parameter vector $\boldsymbol\xi \triangleq [\boldsymbol\xi_1^{\mathrm{T}},\boldsymbol\xi_2^{\mathrm{T}}]^{\mathrm{T}}$ with corresponding FIM 
\begin{equation}
\vec{J}_{\boldsymbol\xi} \triangleq
\left[
\begin{array}{cc}
	\vec{J}_{\boldsymbol\xi_1\boldsymbol\xi_1} & \vec{J}_{\boldsymbol\xi_1\boldsymbol\xi_2}\\
	\vec{J}_{\boldsymbol\xi_1\boldsymbol\xi_2}^{\mathrm{T}} & \vec{J}_{\boldsymbol\xi_2\boldsymbol\xi_2}
\end{array}
\right],
\label{eq:EFIM_helper}
\end{equation}
the EFIM of $\boldsymbol\xi_1$ is obtained by 
\begin{equation}
\vec{J}_{\boldsymbol\xi_1}^{\mathrm{e}} \triangleq \vec{J}_{\boldsymbol\xi_1\boldsymbol\xi_1}-\vec{J}_{\boldsymbol\xi_1\boldsymbol\xi_2}\vec{J}_{\boldsymbol\xi_2\boldsymbol\xi_2}^{-1} \vec{J}_{\boldsymbol\xi_1\boldsymbol\xi_2}^{\mathrm{T}}.
\label{eq:def_EFIM}
\end{equation}
Intuitively, the fact that the parameters of $\boldsymbol\xi_2$ are not perfectly known, leads to a loss in information which is quantified by $\vec{J}_{\boldsymbol\xi_1\boldsymbol\xi_2}\vec{J}_{\boldsymbol\xi_2\boldsymbol\xi_2}^{-1} \vec{J}_{\boldsymbol\xi_1\boldsymbol\xi_2}^{\mathrm{T}}$. 
\end{definition}
Using \eqref{eq:J_k} and \eqref{eq:def_EFIM}, the EFIM of $\vec{J}_{\theta_{\mathrm{TX},k}\vec{h}_k}$ with respect to $\theta_{\mathrm{TX},k}$ is given by 
\begin{equation}
\vec{J}_{\theta_{\mathrm{TX},k}}^{\mathrm{e}} = \frac{1}{\tilde{\sigma}_{\theta_{\mathrm{TX,k}}}^2}-\left(b_{\mathrm{R,k}}^2\sigma_{h_{\mathrm{R,k}}}^2+b_{\mathrm{I,k}}^2\sigma_{h_{\mathrm{I,k}}}^2\right)\triangleq \frac{1}{\sigma_{\theta_{\mathrm{TX,k}}}^2}.
\label{eq:EFIM_theta_TX}
\end{equation}
Hence the (E)FIM\footnote{We slightly abuse the notation here since we denote the EFIM ${\vec{J}}_{\bar{\boldsymbol\eta}_k}^{\mathrm{e}}$ by ${\vec{J}}_{\bar{\boldsymbol\eta}_k}$ and call it FIM to avoid confusion with another EFIM later on.} of the $k^{\text{th}}$ path is given by ${\vec{J}}_{\bar{\boldsymbol\eta}_k} \triangleq \mathrm{diag} \left(1/\sigma_{\tau_{\mathrm{k}}}^2 ,1/\sigma_{\theta_{\mathrm{TX,k}}}^2,1/\sigma_{\theta_{\mathrm{TX,k}}}^2\right)$, where $\bar{\boldsymbol\eta}_k= [\tau_{\mathrm{k}}, \theta_{\mathrm{TX,k}},\theta_{\mathrm{\mathrm{RX},k}}]^{\mathrm{T}}$. Consequently, the FIM of all paths is given by
\begin{equation}
{\vec{J}}_{\bar{\boldsymbol\eta}} = \mathrm{blkdiag}\left({\vec{J}}_{\bar{\boldsymbol\eta}_1},...,{\vec{J}}_{\bar{\boldsymbol\eta}_{K-1}}\right).
\label{eq:FIM_diagonal_eta_bar}
\end{equation}
Each $\sigma^2$-term reflects the quality of the parameter estimation of the respective parameter, e.g. large $\sigma^2_{\tau_0}$ means that the TOA of the first path cannot be estimated accurately. Note that the $\sigma^2$-terms in \eqref{eq:FIM_diagonal_eta_bar} depend on the number of antennas, beamforming, bandwidth, and receiver location, as evaluated and analyzed in \cite{SZASW2017}.

Finally, we transform the FIM in \eqref{eq:FIM_diagonal_eta_bar} to the position, orientation, and point of incidence domain using the geometric relationships between the channel parameters $\bar{\boldsymbol\eta}=[\bar{\boldsymbol\eta}_0,...,\bar{\boldsymbol\eta}_{K-1}]$ and $\tilde{\boldsymbol\eta}\triangleq  [\vec{p}^{\mathrm{T}},\alpha,\vec{s}_1^{\mathrm{T}},...,\vec{s}_{K-1}^{\mathrm{T}}]^{\mathrm{T}}$. In particular, the FIM of the position-related parameters is given by \cite{K1993}
\begin{equation}
\vec{J}_{\tilde{\boldsymbol\eta}} \triangleq \vec{T}\vec{J}_{\bar{\boldsymbol\eta}} \vec{T}^{\mathrm{T}},
\label{eq:J_eta_tilde}
\end{equation}
where $\vec{T}\triangleq \frac{\partial\bar{\boldsymbol\eta}^{\mathrm{T}}}{\partial\tilde{\boldsymbol\eta}}$. 

\textit{Remark}: We normalize the units of the position-related and orientation-related entries in the transformation matrix by $1$ m and $1$ rad, respectively, in order to avoid a mixed-units FIM of the position, orientation, and points of incidence after the transformation. After the (normalized) transformation, the FIM of the position, orientation, and points of incidence will be dimensionless, which allows us to employ the standard inner product (dot product) to define the norm of a vector. For notational convenience, we omit the normalization constants ($1$ m and $1$ rad).
\begin{lemma}
The transformation matrix $\vec{T}$ is an upper triangle block-matrix
\begin{equation}
\vec{T} \triangleq
\left[
\begin{array}{c|cccc}
\vec{T}_{\vec{P}}^{(0)} & 
\vec{T}_{\vec{P}}^{(1)} & 
\vec{T}_{\vec{P}}^{(2)} & 
\cdots &
\vec{T}_{\vec{P}}^{(K-1)} \\
  \hline

\vec{0}_{2\times3} & 
\vec{T}_{\vec{S}_1} & 
\vec{0}_{2\times3} & 
\cdots& 
\vec{0}_{2\times3} \\
\vec{0}_{2\times3} & 
\vec{0}_{2\times3}  & 
\vec{T}_{\vec{S}_2} & 
\ddots &
\vec{0}_{2\times3} \\

\vdots & 
\vdots & 
\ddots& 
\ddots &
\vdots \\

\vec{0}_{2\times3}  & 
\vec{0}_{2\times3}  & 
\vec{0}_{2\times3} & 
\cdots&
\vec{T}_{\vec{S}_{K-1}}\\
\end{array}
\right]
\triangleq
\left[
\begin{array}{c|c}
\vec{A} & \vec{B} \\
\hline
\vec{0} & \vec{D}
\end{array}
\right],
\label{eq:T_matrix}
\end{equation}
where \cite{SGDSW2017}
\begin{equation}
\vec{T}_{\vec{P}}^{(k)}=
\left[
\begin{array}{ccc}
\partial\tau_k/\partial\vec{p} & \partial\theta_{\mathrm{TX},k}/\partial\vec{p} &\partial\theta_{\mathrm{RX},k}/\partial\vec{p} \\ 
\partial\tau_k/\partial\alpha & \partial\theta_{\mathrm{TX},k}/\partial\alpha &\partial\theta_{\mathrm{RX},k}/\partial\alpha \\ 
\end{array}
\right]
\label{eq:P_P_i}
\end{equation}
and 
\begin{equation}
\vec{T}_{\vec{s}_k}=
\left[
\begin{array}{ccc}
\partial\tau_k/\partial\vec{s}_k & \partial\theta_{\mathrm{TX},k}/\partial\vec{s}_k &\partial\theta_{\mathrm{RX},k}/\partial\vec{s}_k \\ 
\end{array}
\right].
\label{eq:P_P_sk}
\end{equation}
\label{lemma:T}
The individual elements are summarized in Appendix \ref{sec:appendix_A}. 
\end{lemma}

\begin{lemma}
\label{lemma:J_p_alpha}
The FIM of the position-related parameters is given by
\begin{equation}
\vec{J}_{\tilde{\boldsymbol\eta}} =
\left[
\begin{array}{c|c}
\vec{A}[\vec{J}_{{\bar{\boldsymbol\eta}}}]_{1:3,1:3}\vec{A}^{\mathrm{T}} +\vec{B}[\vec{J}_{{\bar{\boldsymbol\eta}}}]_{4:3K,4:3K}\vec{B}^{\mathrm{T}} &
\vec{B}[\vec{J}_{{\bar{\boldsymbol\eta}}}]_{4:3K,4:3K}\vec{D}^{\mathrm{T}} \\
\hline

\vec{D}[\vec{J}_{{\bar{\boldsymbol\eta}}}]_{4:3K,4:3K}\vec{B}^{\mathrm{T}}&
\vec{D}[\vec{J}_{{\bar{\boldsymbol\eta}}}]_{4:3K,4:3K}\vec{D}^{\mathrm{T}} \\
\end{array}
\right].
\label{eq:FIM_pos}
\end{equation}
\label{lemma:FIM_pos}
\end{lemma}

\begin{proof}
Evaluate \eqref{eq:J_eta_tilde} using Lemma \ref{lemma:T}.
\end{proof}
Defining $\tilde{\boldsymbol\eta}_{\vec{p},\alpha}\triangleq[\vec{p}^{\mathrm{T}},\alpha]^{\mathrm{T}}$, we obtain the EFIM of the position and orientation using \eqref{eq:def_EFIM} and \eqref{eq:FIM_pos} from  Lemma \ref{lemma:FIM_pos}
\begin{equation}
\begin{array}{cl}
\vec{J}_{\tilde{\boldsymbol\eta}_{\vec{p},\alpha}}^{\mathrm{e}}=&\underbrace{\vec{A}[\vec{J}_{{\bar{\boldsymbol\eta}}}]_{1:3,1:3}\vec{A}^{\mathrm{T}}}_{\triangleq\tilde{\vec{A}}^{\mathrm{(G)}}\text{ - LOS info gain}} +\underbrace{\vec{B}[\vec{J}_{{\bar{\boldsymbol\eta}}}]_{4:3K,4:3K}\vec{B}^{\mathrm{T}}}_{\triangleq\tilde{\vec{B}}^{\mathrm{(G)}}\text{ - NLOS info gain}} \\
																	&-\underbrace{\vec{B}[\vec{J}_{{\bar{\boldsymbol\eta}}}]_{4:3K,4:3K}\vec{D}^{\mathrm{T}} (\vec{D}[\vec{J}_{{\bar{\boldsymbol\eta}}}]_{4:3K,4:3K}\vec{D}^{\mathrm{T}})^{-1}\vec{D}[\vec{J}_{{\bar{\boldsymbol\eta}}}]_{4:3K,4:3K}\vec{B}^{\mathrm{T}}}_{\triangleq\tilde{\vec{B}}^{\mathrm{(L)}}\text{ - NLOS info loss}}.
\end{array}
\label{eq:eq_FIM}
\end{equation}
It becomes evident that the EFIM is composed of three terms. The first term $\tilde{\vec{A}}^{\mathrm{(G)}}$ quantifies the information gain from the LOS path. The second term $\tilde{\vec{B}}^{\mathrm{(G)}}$ quantifies the information gain from the NLOS components. Finally, the third term $\tilde{\vec{B}}^{\mathrm{(L)}}$ specifies the loss of information which accounts for the fact that the points of incident of NLOS paths are unknown. Considering the structure of $\vec{J}_{\bar{\boldsymbol\eta}}$ and $\vec{T}$ the NLOS info gain and loss can be written as
\begin{equation}
\tilde{\vec{B}}^{\mathrm{(G)}}=\sum_{k=1}^{K-1} \vec{T}_{\vec{P}}^{(k)} \vec{J}_{\bar{\boldsymbol\eta}_k}\left(\vec{T}_{\vec{P}}^{(k)}\right)^{\mathrm{T}}
\label{eq:MP_gain_as_sum}
\end{equation}
and
\begin{equation}
\tilde{\vec{B}}^{\mathrm{(L)}}=\sum_{k=1}^{K-1} \vec{T}_{\vec{P}}^{(k)} \vec{J}_{\bar{\boldsymbol\eta}_k}\left(\vec{T}_{\vec{s}_k}\right)^{\mathrm{T}} \left(\vec{T}_{\vec{s}_k} \vec{J}_{\bar{\boldsymbol\eta}_k}\left(\vec{T}_{\vec{s}_k}\right)^{\mathrm{T}} \right)^{-1} \vec{T}_{\vec{s}_k} \vec{J}_{\bar{\boldsymbol\eta}_k}\left(\vec{T}_{\vec{P}}^{(k)}\right)^{\mathrm{T}},
\label{eq:MP_loss_as_sum}
\end{equation}
respectively. 

\textit{Interpretation}:
The term in \eqref{eq:MP_gain_as_sum} reflects the gain of information from NLOS components if the positions of the corresponding points of incident $\vec{s}_k$ were perfectly known. Since the position of each point of incidence has to be estimated, as well, the uncertainty regarding the position of the scatterer leads to a loss of information which is quantified in \eqref{eq:MP_loss_as_sum}. 

In the following, we decompose the terms $\vec{A}^{\mathrm{(G)}}$ and $\vec{B}^{\mathrm{(N)}}=\vec{B}^{\mathrm{(G)}}-\vec{B}^{\mathrm{(L)}}$ with the goal to express the EFIM in \eqref{eq:eq_FIM} as the sum of rank one matrices, where each rank one matrix is given by the outer-product of the unit-norm eigenvector $\vec{v}_k$ corresponding to the only non-zero eigenvalue $\lambda_k$ of the rank one matrix, i.e. $\vec{J}_{\tilde{\boldsymbol\eta}_{\vec{p},\alpha}}^{\mathrm{e}} = \sum_{k} \lambda_k \vec{v}_k \vec{v}_k^{\mathrm{T}}$. We show that each rank one matrix corresponding to a NLOS path is positive semi-definite meaning that each NLOS path improves the position and orientation estimation accuracy. 

For notational convenience, we define the following matrix template:
\begin{equation}
\boldsymbol\Upsilon_{n,m}(\theta,\phi,\rho)\triangleq
\left[
\begin{array}{ccc}
	\cos^2(\theta+\phi) & (-1)^{n}\sin(\theta)\cos(\theta) & (-1)^{m}\rho \sin(\theta) \\
	(-1)^{n}\sin(\theta)\cos(\theta) & \sin^2(\theta+\phi) & (-1)^{m+1}\rho \cos(\theta) \\
	(-1)^{m}\rho\sin(\theta) & (-1)^{m+1}\rho \cos(\theta) & \rho^2 
\end{array}
\right].
\label{eq:matrix_template}
\end{equation}
\subsection{LOS Information Gain} 
Using \eqref{eq:T_matrix}, the results from Appendix \ref{sec:appendix_A}, and by simple algebra, the first term $\vec{A}[\vec{J}_{{\boldsymbol\eta}}]_{1:3,1:3}\vec{A}^{\mathrm{T}}$ can be easily shown to be 
\begin{equation}
\begin{array}	{cl}
	\tilde{\vec{A}}^{\mathrm{(G)}}=& \underbrace{\frac{1}{\sigma^2_{\tau_0}c^2} \boldsymbol\Upsilon_{0,0}(\theta_{\mathrm{TX},0},0,0)}_{\tilde{\vec{A}}_R^{\mathrm{(G)}}}
	+\underbrace{\frac{1}{\sigma^2_{\theta_{\mathrm{TX},0}}\left\|\vec{p}-\vec{q}\right\|^2} \boldsymbol\Upsilon_{1,0}(\theta_{\mathrm{TX},0},\pi/2,0)}_{\tilde{\vec{A}}_D^{\mathrm{(G)}}}\\
	&+\underbrace{\frac{1}{\sigma^2_{\theta_{\mathrm{RX},0}}\left\|\vec{p}-\vec{q}\right\|^2} \boldsymbol\Upsilon_{1,0}(\theta_{\mathrm{TX},0},\pi/2,\left\|\vec{p}-\vec{q}\right\|)}_{\tilde{\vec{A}}_A^{\mathrm{(G)}}}.
\end{array}
\label{eq:EFIM_gain_LOS}
\end{equation}
\begin{proposition}
\label{prop:LOS}
The eigenvalues of the matrices $\tilde{\vec{A}}_R^{\mathrm{(G)}},\tilde{\vec{A}}_D^{\mathrm{(G)}}$, and $\tilde{\vec{A}}_A^{\mathrm{(G)}}$ are given by 
\begin{subequations}
\begin{alignat}{3}
\lambda_{R,0}^{\mathrm{(G)}}&=\frac{1}{\sigma^2_{\tau_0}c^2},
\label{eq:LOS_eval_R}
\\
\lambda_{D,0}^{\mathrm{(G)}}&=\frac{1}{\sigma^2_{\theta_{\mathrm{TX},0}}\left\|\vec{p}-\vec{q}\right\|^2},
\label{eq:LOS_eval_D}
\\
\lambda_{A,0}^{\mathrm{(G)}}&=\frac{\left\|\vec{p}-\vec{q}\right\|^2+1}{\sigma^2_{\theta_{\mathrm{RX},0}}\left\|\vec{p}-\vec{q}\right\|^2},
\label{eq:LOS_eval_A}
\end{alignat}
\label{eq:LOS_eval}
\end{subequations}
respectively. The corresponding eigenvectors are given by 
\begin{subequations}
\begin{alignat}{3}
\vec{v}_{R,0}^{\mathrm{(G)}}&=[\cos(\theta_{\mathrm{TX},0}), \sin(\theta_{\mathrm{TX},0}),0]^{\mathrm{T}},
\label{eq:LOS_evec_R}
\\
\vec{v}_{D,0}^{\mathrm{(G)}}&=[\sin(\theta_{\mathrm{TX},0}), \cos(\theta_{\mathrm{TX},0}),0]^{\mathrm{T}},
\label{eq:LOS_evec_D}
\\
\vec{v}_{A,0}^{\mathrm{(G)}}&=v_{A,0}^{\mathrm{(G)}}\left[\frac{\sin(\theta_{\mathrm{TX},0})}{\left\|\vec{p}-\vec{q}\right\|}, -\frac{\cos(\theta_{\mathrm{TX},0})}{\left\|\vec{p}-\vec{q}\right\|},1\right]^{\mathrm{T}},
\label{eq:LOS_evec_A}
\end{alignat}
\end{subequations}
where $v_{A,0}^{\mathrm{(G)}}=\sqrt{1/(1+\left\|\vec{p}-\vec{q}\right\|^2)}$ ensures that $\vec{v}_{A,0}^{\mathrm{(G)}}$ to is unit-norm.
It is immediately obvious that all eigenvalues are positive.
\end{proposition}
\begin{proof}
See Appendix \ref{sec:appendix_B}.
\end{proof}
\textit{Interpretation}: TOA, AOD, and AOA of the LOS path contribute information to the EFIM. Each of these quantities provides information in \textit{one} direction, where the direction is given by the eigenvector corresponding to the non-zero eigenvalue. TOA and AOD provide position information in orthogonal directions. So the position is identifiable if TOA and AOD can be estimated. Neither TOA, nor AOD provide orientation information since the $\alpha$ component in \eqref{eq:LOS_evec_R} and \eqref{eq:LOS_evec_D} is zero, i.e. for known position, AOD does not contribute to the orientation estimation, but uncertainty in the AOD reduces the orientation EFIM. When $\sigma^2_{\theta_{\mathrm{TX},0}}\rightarrow \infty$, but $\sigma^2_{\theta_{\mathrm{RX},0}}$ and $\sigma^2_{\tau_0}$  are finite, then neither position nor orientation can be estimated. On the other hand, with $\sigma^2_{\theta_{\mathrm{RX},0}}\rightarrow \infty$, the orientation cannot be estimated, while the position can still be estimated. Similarly, $\sigma^2_{\tau_{0}}\rightarrow \infty$, means that the position cannot be estimated, but the orientation can still be determined. From \eqref{eq:LOS_evec_A}, we deduce that AOA provides both position and orientation information. Considering \eqref{eq:LOS_eval_D} and \eqref{eq:LOS_eval_A}, we see that the information gain reduces as the separation between base station and mobile terminal increases. For a large separation, AOA only provides orientation information since the x-y components of the eigenvector in \eqref{eq:LOS_evec_A} go to zero. In that case, the amount of orientation information is given by $1/\sigma^2_{\theta_{\mathrm{RX},0}}$.


\subsection{NLOS Information Gain}
Similar to the previous section, we decompose the information matrices of reflected paths using \eqref{eq:T_matrix} and the results from Appendix \ref{sec:appendix_A}. 
\begin{lemma}
\label{lemma:mp_gain}
The NLOS information gain matrix is given by
\begin{equation}
\tilde{\vec{B}}^{\mathrm{(G)}} =
\sum_{k=1}^{K-1}\underbrace{\frac{1}{\sigma^2_{\tau_k}c^2} \boldsymbol\Upsilon_{0,0}(\theta_{\mathrm{RX},k},0,0)}_{\tilde{\vec{B}}_{R,k}^{\mathrm{(G)}}}
	+\underbrace{\frac{1}{\sigma^2_{\theta_{\mathrm{RX},k}}\left\|\vec{p}-\vec{s}_k\right\|^2} \boldsymbol\Upsilon_{1,1}(\theta_{\mathrm{RX},k},\pi/2,\left\|\vec{p}-\vec{s}_k\right\|)}_{\tilde{\vec{B}}_{A,k}^{\mathrm{(G)}}}.
\label{eq:EFIM_gain_MP}
\end{equation}
\end{lemma}
\begin{proof}
See Appendix \ref{sec:appendix_C_MP_gain}.
\end{proof}

Note that the matrices $\tilde{\vec{B}}_{R,k}^{\mathrm{(G)}}$ and $\tilde{\vec{B}}_{A,k}^{\mathrm{(G)}}$ have the similar structure as $\tilde{\vec{A}}_{R}^{\mathrm{(G)}}$ and $\tilde{\vec{A}}_{A}^{\mathrm{(G)}}$, respectively. Hence the eigenvalues and eigenvectors have also similar structure, and are given by 
\begin{subequations}
\begin{alignat}{2}
\lambda_{R,k}^{\mathrm{(G)}}&=\frac{1}{\sigma^2_{\tau_k}c^2},
\label{eq:MP_eval_R}
\\
\lambda_{A,k}^{\mathrm{(G)}}&=\frac{\left\|\vec{p}-\vec{s}_k\right\|^2+1}{\sigma^2_{\theta_{\mathrm{RX},k}}\left\|\vec{p}-\vec{s}_k\right\|^2},
\label{eq:MP_eval_A}
\end{alignat}
\end{subequations}
and 
\begin{subequations}
\begin{alignat}{2}
\vec{v}_{R,k}^{\mathrm{(G)}}&=[\cos(\theta_{\mathrm{RX},k}), \sin(\theta_{\mathrm{RX},k}),0]^{\mathrm{T}},
\label{eq:MP_evec_R}
\\
\vec{v}_{A,k}^{\mathrm{(G)}}&=v_{A,k}^{\mathrm{(G)}}\left[-\frac{\sin(\theta_{\mathrm{RX},k})}{\left\|\vec{p}-\vec{s}_k\right\|}, \frac{\cos(\theta_{\mathrm{RX},k})}{\left\|\vec{p}-\vec{s}_k\right\|},1\right]^{\mathrm{T}},
\label{eq:MP_evec_A}
\end{alignat}
\end{subequations}
respectively. The term  $v_{A,k}^{\mathrm{(G)}}=\sqrt{1/(1+\left\|\vec{p}-\vec{s}_k\right\|^2)}$ normalizes the eigenvector.

\textit{Interpretation}: AOA and TOA of each NLOS component contribute to the EFIM. The eigenvalues can be interpreted as the achievable information gains if the points of incidence were perfectly known. From \eqref{eq:MP_evec_R}, we see that TOA provides only position information. From \eqref{eq:MP_evec_A}, we deduce that AOA provides both position and orientation information. The term in \eqref{eq:MP_eval_A} suggests that the information gain of NLOS components decreases as the separation between the mobile terminal and point of incidence increases. For large distances, the information gain converges to the inverse of the AOA estimation accuracy $1/\sigma^2_{\mathrm{RX},k}$. Obverse that in this case, a NLOS component provides only mostly orientation information gain, since the x-y components of the eigenvector in \eqref{eq:MP_evec_A} go to zero.  Note that AOD estimation does not provide any information gain. However, we will see in the next section that the quality of the AOD estimation strongly influences the information loss.

\subsection{NLOS Information Loss}
The NLOS information loss matrix is not as obvious to decompose as the previous decompositions since it involves multiple matrix-matrix products including the inverse of a matrix-matrix product. In the following, we will present the results of the decomposition. 
\begin{lemma}
The NLOS information loss matrix is given by
\begin{equation}
\tilde{\vec{B}}^{\mathrm{(L)}}
 = \sum_{k=1}^{K-1}\underbrace{w_{R,k}^{\mathrm{(L)}}\boldsymbol\Upsilon_{0,0}(\theta_{\mathrm{TX},0},0)}_{\tilde{\vec{B}}_{R,k}^{\mathrm{(L)}}}+\underbrace{w_{A,k}^{\mathrm{(L)}}\boldsymbol\Upsilon_{1,1}(\theta_{\mathrm{RX},k},\left\|\vec{p}-\vec{q}\right\|)}_{\tilde{\vec{B}}_{A,k}^{\mathrm{(L)}}}-\gamma_{\vec{s}_k} \vec{B}_{k}^{\mathrm{(L)}}
\label{eq:EFIM_loss_MP}
\end{equation}
where the weights\footnote{Each weight is a function of $(\theta_{\mathrm{TX},k},\theta_{\mathrm{RX},k},\tau_k,\sigma^2_{\theta_{\mathrm{TX},k}},\sigma^2_{\theta_{\mathrm{RX},k}},\sigma^2_{\tau_k},\vec{p},\vec{q},\vec{s}_k)$. To simplify notation, we dropped the arguments.} $w_{R,k}^{\mathrm{(L)}}, w_{A,k}^{\mathrm{(L)}}$ and $\gamma_{\vec{s}_k}$, and the matrix $\vec{B}_{k}^{\mathrm{(L)}}$  are defined in Appendix \ref{sec:appendix_C_MP_loss}.
\end{lemma}
\begin{proof}
See Appendix \ref{sec:appendix_C_MP_loss}. 
\end{proof}

\textit{Interpretation}: The information loss of each path is composed of three terms.
Comparing \eqref{eq:EFIM_gain_MP} and \eqref{eq:EFIM_loss_MP}, we see that the first two terms  of every path $k$ have the same structure. In particular, the TOA information gain $\tilde{\vec{B}}_{R,k}^{\mathrm{(G)}}$ and loss $\tilde{\vec{B}}_{R,k}^{\mathrm{(L)}}$ only differ in the weights. The same holds for the AOA information gain $\tilde{\vec{B}}_{A,k}^{\mathrm{(G)}}$ and loss $\tilde{\vec{B}}_{A,k}^{\mathrm{(L)}}$. Hence each information gain and loss pair has eigenvectors that are pointing in the same direction. Yet, the eigenvalues have opposite signs. For these terms, net information gains are immediately given by the differences of the respective weights. On the other hand, the third term has eigenvectors which are not aligned with any other eigenvectors. In the following subsection, we will show that combining the gain and loss matrices of every NLOS component results in a net information gain matrix which has only a single non-zero, positive eigenvalue. 
\subsection{Net NLOS Information Gain}
\label{subsec:eff_MP_gain}
\begin{lemma}
\label{lemma:mp_eff_gain}
 The net NLOS gain matrix $\tilde{\vec{B}}^{(\mathrm{N})}\triangleq\tilde{\vec{B}}^{\mathrm{(G)}}-\tilde{\vec{B}}^{\mathrm{(L)}}$ is given by
\begin{equation}
\tilde{\vec{B}}^{(\mathrm{N})}
 =\sum_{k=1}^{K-1} \epsilon_{\mathbf{s}_k} \boldsymbol\Upsilon_{0,0}(\theta_{\mathrm{RX},0},0,0)+\beta_{\mathbf{s}_k}\boldsymbol\Upsilon_{1,1}(\theta_{\mathrm{RX},k},\pi/2,\left\|\vec{p}-\vec{q}\right\|)
	+ \gamma_{\mathbf{s}_k} \vec{B}_{k}^{\mathrm{(L)}},
\label{eq:EFIM_eff_Gain}
\end{equation}
where $\epsilon_{\mathbf{s}_k}\triangleq \frac{1}{\sigma_{\tau_k}^2c^2}-w_{R,k}^{\mathrm{(L)}}$ and $\beta_{\mathbf{s}_k}\triangleq \frac{1 }{\sigma_{\theta_{\mathrm{RX},k}}^2\left\|\mathbf{p}-\mathbf{s}_k\right\|^2}-w_{A,k}^{\mathrm{(L)}}$.
\end{lemma}
\begin{proof}
Applying Lemma \ref{lemma:mp_gain} and \ref{lemma:mp_eff_gain}, we immediately obtain \eqref{eq:EFIM_eff_Gain}.
\end{proof}

We define $\tilde{\vec{B}}^{(N)}\triangleq\sum_{k=1}^{K-1}\boldsymbol\Psi_{\mathbf{s}_k}$, where 
\begin{equation}
\boldsymbol\Psi_{\mathbf{s}_k}\triangleq\epsilon_{\mathbf{s}_k} \boldsymbol\Upsilon_{0,0}(\theta_{\mathrm{TX},0},0,0)+\beta_{\mathbf{s}_k}\boldsymbol\Upsilon_{1,1}(\theta_{\mathrm{RX},k},\pi/2,\left\|\vec{p}-\vec{q}\right\|)
	+ \gamma_{\mathbf{s}_k} \vec{B}_{k}^{\mathrm{(L)}}.
\label{eq:Psi_k}
\end{equation}
\begin{theorem}
\label{theorem:net_gain}
The net information gain matrix of the $k^{\mathrm{th}}$ NLOS component $\boldsymbol\Psi_{\mathbf{s}_k}$ is rank one. The only non-zero eigenvalue of $\boldsymbol\Psi_{\mathbf{s}_k}$ is always positive and given by
\begin{equation}
\lambda_{\mathbf{s}_k}
\triangleq
\frac{2+\left\|\mathbf{p}-\mathbf{s}_k\right\|^2 (1+\cos(\Delta\theta_k))}
{(1-\cos(\Delta\theta_k))c^2\sigma^2_{\tau_k}
+(1+\cos(\Delta\theta_k))(\left\|\mathbf{p}-\mathbf{s}_k\right\|^2\sigma^2_{\theta_{\mathrm{RX},k}}+\left\|\mathbf{q}-\mathbf{s}_k\right\|^2\sigma^2_{\theta_{\mathrm{TX},k}})},
\label{eq:non_zero_eigenvalue}
\end{equation}
with the corresponding unit-norm eigenvector 
\begin{equation}
\mathbf{v}_{\mathbf{s}_k}
\triangleq
v_{\vec{s}_k}
\begin{bmatrix}
-\frac{1}{\left\|\mathbf{p}-\mathbf{s}_k\right\|}\left(\frac{\epsilon_{\mathbf{s}_k}}{\gamma_{\mathbf{s}_k}}\cos(\theta_{\mathrm{RX},k})+ \sin(\theta_{\mathrm{RX},k})\right)
\\
\frac{1}{\left\|\mathbf{p}-\mathbf{s}_k\right\|}\left(-\frac{\epsilon_{\mathbf{s}_k}}{\gamma_{\mathbf{s}_k}}\sin(\theta_{\mathrm{RX},k})+ \cos(\theta_{\mathrm{RX},k})\right)
\\
1
\end{bmatrix},
\label{eq:non_zero_eigenvector}
\end{equation}
where $\Delta\theta_k \triangleq \theta_{\mathrm{RX},k}-\theta_{\mathrm{TX},k}$ and $v_{\vec{s}_k}=\sqrt{(1+\cos(\Delta\theta_k))/(2+\left\|\mathbf{p}-\mathbf{s}_k\right\|^2 (1+\cos(\Delta\theta_k)))}$.
\end{theorem}
\begin{proof}
Outline: using Lemma \ref{lemma:mp_eff_gain}, it can be shown that $\boldsymbol\Psi_{\mathbf{s}_k}$ is rank one. Hence the only non-zero eigenvalue is given by $\lambda_{\mathbf{s}_k}=\tr\left(\boldsymbol\Psi_{\mathbf{s}_k}\right)$. Positivity of $\lambda_{\mathbf{s}_k}$ is readily seen since all terms in \eqref{eq:non_zero_eigenvalue} are non-negative, i.e. $(1+\cos(\theta_{\mathrm{RX},k}-\theta_{\mathrm{TX},k}))\geq 0$, $(1-\cos(\theta_{\mathrm{RX},k}-\theta_{\mathrm{TX},k}))\geq 0$, and all other terms are positive by definition. Finally, we show that the vector $\mathbf{v}_{\mathbf{s}_k}$ in \eqref{eq:non_zero_eigenvector} is in the null space of $\left(\lambda_{\mathbf{s}_k}\vec{I}-\boldsymbol\Psi_{\mathbf{s}_k} \right)$. Hence $\mathbf{v}_{\mathbf{s}_k}$ is the corresponding eigenvector. The details can be found in Appendix \ref{sec:appendix_C_MP_eff_gain}.
\end{proof}

\textit{Interpretation}: 
Each NLOS component provides one dimensional Fisher information for all three parameters ($p_x$,$p_y$, and $\alpha$). The information of all NLOS components is additive and, thus, contributes to the EFIM which, in turn, reduces the position and orientation error bound. Hence NLOS components can be harnessed to increase the position and orientation estimation accuracy in 5G mmWave MIMO systems. The amount of information gained from a NLOS component depends strongly on the geometry (this can be seen from the $\cos(\Delta\theta_k)$ terms in \eqref{eq:non_zero_eigenvalue}), i.e. points of incidence which are located in certain areas provide more information than points of incidence in other areas. From the denominator in \eqref{eq:non_zero_eigenvalue}, it can be seen that NLOS components provide considerable Fisher information if and only if TOA, AOD, and AOA can be estimated sufficiently accurate. If any of the three quantities cannot be estimated, e.g. AOD, this is equivalent to $\sigma^2_{\theta_{\mathrm{TX},k}}\rightarrow \infty$. Hence $\lambda_{\vec{s}_k}\rightarrow 0$ which means that the NLOS component does not provide information useful for position and orientation estimation. Finally, we see that points of incidence which are close to the base station or the mobile terminal (i.e. with small $\left\|\mathbf{q}-\mathbf{s}_k\right\|$ or $\left\|\mathbf{p}-\mathbf{s}_k\right\|$, respectively) are more informative than points of incidence which are farther away. Intuitively, this observation makes sense because, for a given angular estimation accuracy, the position of a point of incidence can be estimated better if it is close the base station. If it is far away from the base station, a small angular estimation error translates into a large estimation error of the point of incidence. The same holds for the path from the point of incidence to the mobile terminal. 
\begin{corollary}
\label{corol:sum_of_rank_one}
The EFIM of position and orientation is given by the sum of the outer-products of the eigenvectors weighted by the corresponding eigenvalues, i.e.
\begin{equation}
\vec{J}_{\tilde{\boldsymbol\eta}_{\vec{p},\alpha}}^{\mathrm{e}}=\sum_{j\in R,D,A} \lambda_{j,0}^{\mathrm{(G)}}\vec{v}_{j,0}^{\mathrm{(G)}}\left(\vec{v}_{j,0}^{\mathrm{(G)}}\right)^{\mathrm{T}}+\sum_{k=1}^{K-1} \lambda_{\mathbf{s}_k}\vec{v}_{\vec{s}_k}\vec{v}_{\vec{s}_k}^{\mathrm{T}}.
\label{eq:EFIM_corollary}
\end{equation}
\end{corollary}
\begin{proof}
Considering Lemma \ref{prop:LOS} and Theorem \ref{theorem:net_gain}, we know that each matrix in \eqref{eq:eq_FIM} is rank one. Since any rank one matrix can be written as the outer-product of the unit-norm eigenvector with itself weighted by the only non-zero eigenvalue, \eqref{eq:EFIM_corollary} follows. 
\end{proof}

\textit{Interpretation:}
Note that the position and orientation of the receiver can be determined if $\vec{J}_{\tilde{\boldsymbol\eta}_{\vec{p},\alpha}}^{\mathrm{e}}$ is non-singular. From \eqref{eq:EFIM_corollary}, we observe that even in the absence of the LOS component, unambiguous position and orientation is estimation is possible if at least three NLOS components with distinct points of incidence contribute to the received signal.

For illustration purposes, we consider the position components and the orientation component of the eigenvalue-eigenvector product $\lambda_{\vec{s}_k}\mathbf{v}_{\mathbf{s}_k}$ separately. This can be interpreted as the projection of the eigenvalue-eigenvector product on the x-y plane (position components) and on the $\alpha$ line (orientation component), respectively. After a few algebraic manipulations, it can be shown that the length of the projected eigenvalue-eigenvector product in the x-y plane is given by
\begin{equation}
\begin{array}{cl}
\tilde{\lambda}_{\vec{s}_k}  &\triangleq  \lambda_{\vec{s}_k} \sqrt{v_{\vec{s}_k,x}^2+v_{\vec{s}_k,y}^2 }
\\
														& =
\frac{2}
{\left(1-\cos(\Delta\theta_k)\right)c^2\sigma^2_{\tau_k}
+\left(1+\cos(\Delta\theta_k)\right)\left(\left\|\mathbf{p}-\mathbf{s}_k\right\|^2\sigma^2_{\theta_{\mathrm{RX},k}}+\left\|\mathbf{q}-\mathbf{s}_k\right\|^2\sigma^2_{\theta_{\mathrm{TX},k}}\right)},
\end{array}
\label{eq:tilde_lambda2}
\end{equation}
where $v_{\vec{s}_k,x}$ and $v_{\vec{s}_k,y}$ denote the first and second component of $\vec{v}_{\vec{s}_k}$, respectively. The length of the projected eigenvalue-eigenvector product along the $\alpha$ line is given by 
\begin{equation}
\begin{array}{cl}
	\bar{\lambda}_{\mathbf{s}_k} &\triangleq  \lambda_{\vec{s}_k}  v_{\vec{s}_k,\alpha}
\\
														&=\frac{\left\|\mathbf{p}-\mathbf{s}_k\right\|^2 (1+\cos(\Delta\theta_k))}
{(1-\cos(\Delta\theta_k))c^2\sigma^2_{\tau_k}
+(1+\cos(\Delta\theta_k))(\left\|\mathbf{p}-\mathbf{s}_k\right\|^2\sigma^2_{\theta_{\mathrm{RX},k}}+\left\|\mathbf{q}-\mathbf{s}_k\right\|^2\sigma^2_{\theta_{\mathrm{TX},k}})},
\end{array}
\label{eq:bar_lambda}
\end{equation}
where $v_{\vec{s}_k,\alpha}$ denotes the third component of $\vec{v}_{\vec{s}_k}$. 

Finally, we define the position error as
\begin{equation}
\mathrm{PEB}= \sqrt{\mathrm{tr}\left\{\left[\left(\vec{J}_{\tilde{\boldsymbol\eta}_{\vec{p},\alpha}}^{\mathrm{e}}\right)^{-1}\right]_{1:2,1:2}\right\}}.
\label{eq:def_PEB}
\end{equation}
\subsection{Discussion and Implications}
Our findings provide insights into the problem of position and orientation estimation in 5G mmWave MIMO. Our results may turn out to be useful for designing position and orientation estimators. For instance, Theorem \ref{theorem:net_gain} revealed that every NLOS path provides position and orientation information, and hence every NLOS path reduces the PEB. To maximize the estimation accuracy, a position and orientation estimator should be designed such that all NLOS paths are considered. On the other hand, we know from Theorem \ref{theorem:net_gain} that the net NLOS information gain among paths can vary strongly, i.e. some NLOS paths provide significantly more information than others. To reduce the complexity of an estimator, paths with insignificant net NLOS information gain could be neglected by an estimator without considerable losses in terms of accuracy. For a given observation, the problem whether to consider or neglect a NLOS path is still an open challenge, which requires further research. Theorem \ref{theorem:net_gain} inherently implies that the position of the point of incidence of a NLOS path can be estimated. Thus, simultaneous localization and mapping (SLAM) can be performed with single transmission burst. Moreover, our results expose some explicit laws, e.g. \eqref{eq:non_zero_eigenvector} shows that a NLOS path, whose point of incidence is far away from the mobile terminal, contains mainly orientation information and only marginal position information. From Corollary \ref{corol:sum_of_rank_one}, it can be seen that in the absence of the LOS path, three NLOS paths can lead to a full rank EFIM. Hence NLOS-only position and orientation estimation are possible with soley a single anchor. 

\section{Numerical Example}
\label{sec:numerical_example}
In this section, we provide numerical examples for position and orientation estimation using mmWave MIMO with 5G-typical parameters. First, we describe the simulation setup. Subsequently, we present two examples, which 1) discuss the reduction of the PEB when a NLOS path complements the LOS path and 2) demonstrate LOS-free, single-anchor localization.  
\subsection{Simulation Setup}
We consider a mobile terminal that is located at $\vec{p}=[5, 5]^{\mathrm{T}}$ with an orientation angle of $\alpha=\pi/2$. In contrast to \cite{SZASW2017}, we consider uniform linear arrays (ULAs) with half-wavelength inter-element spacing at the transmitter and receiver that consist of $N_{\mathrm{TX}}$ and $N_{\mathrm{RX}}=25$ antennas, respectively. The operating carrier frequency is $f=38$ GHz. Regarding the pilot signal, we consider an ideal sinc pulse with $B=125$ MHz, $E_{\mathrm{s}}/T_{\mathrm{s}}=0$ dBm, $N_0=-170$dBm/Hz, and $N_{\mathrm{s}}=16$ symbols. We consider a DFT-based beamforming matrix with $N_{\mathrm{B}}=50$ beams uniformly spaced between $[0 ,\pi)$. The complex channel gains for each path are generated according to a geometric model\cite{G2005}. The gain is proportional to the path loss and the phase is uniformly distributed. In particular, we assume that $\left|h_0\right|^2=(\lambda /4\pi)^2/\left\|\vec{p}-\vec{q}\right\|^2$ for the LOS path and $\left|h_k\right|^2=(\lambda /4\pi)^2\Gamma_R/(\left\|\vec{q}-\vec{s}_k\right\|+\left\|\vec{p}-\vec{s}_k\right\|)^2, k=1,2$ for NLOS paths\footnote{For the numerical examples, we consider the point of incidence of a reflector as the source of the NLOS path.}, where $\Gamma_R=0.7$.
\subsection{Example 1}
\begin{figure*}[t]%
\centering
\hspace*{\fill}%
\begin{subfigure}{.48\columnwidth}
\centering
\includegraphics[width=\columnwidth]{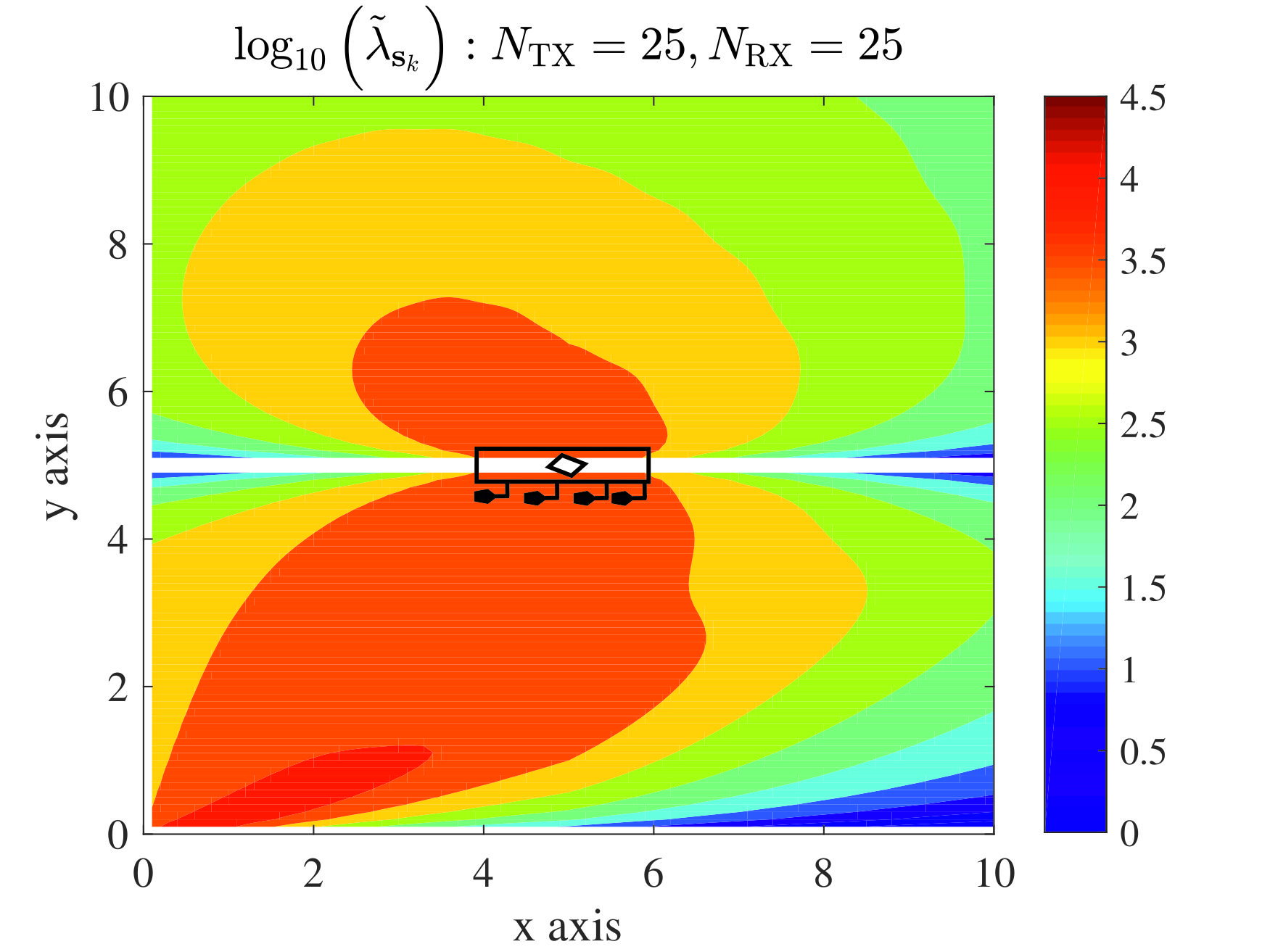}%
\label{fig:pos_info_heatmap_left}%
\end{subfigure}
\begin{subfigure}{.48\columnwidth}
\centering
\includegraphics[width=\columnwidth]{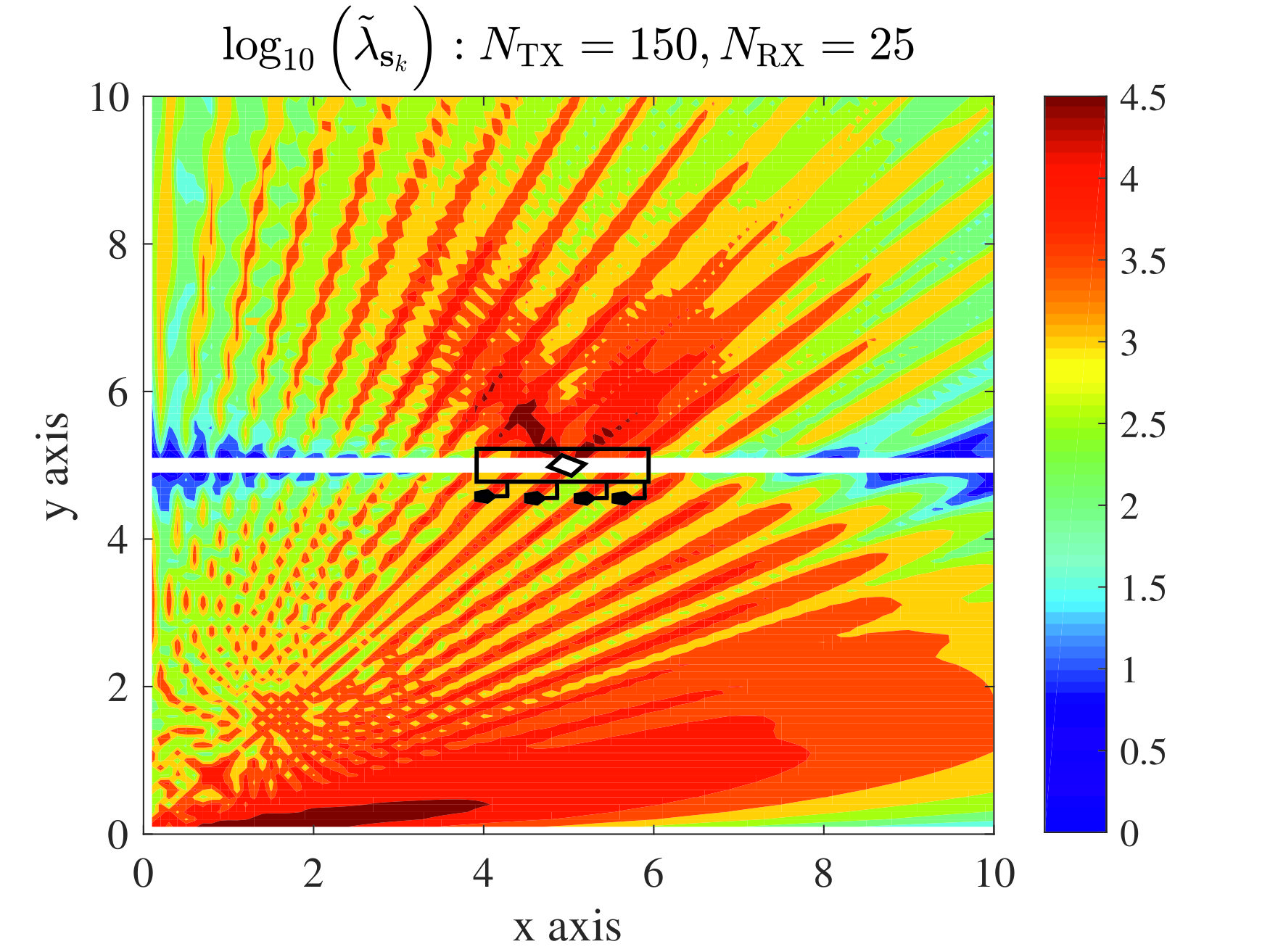}%
\label{fig:pos_info_heatmap_right}%
\end{subfigure}
\caption{\textit{Net position information gain - }  Narrower beamwidth caused by more transmit antennas (right) allows for larger net position information gain when compared to wider beamwidth (left).}
\label{fig:pos_info_heatmap}%
\end{figure*}
We consider the LOS path and one NLOS path in this example. The reflector, which causes the NLOS path, is moved in the x-y plane between ${0<s_{\mathrm{x},1}\leq 10}$ and ${0<s_{\mathrm{y},1}\leq 10}$. For every location of the point of incidence of the reflector $\vec{s}_1$, we determine the net position information gain $\tilde{\lambda}_{\vec{s}_1}$, which is depicted in log-scale in Fig. \ref{fig:pos_info_heatmap}. The array of the mobile terminal is shown in black. We consider different transmit array sizes to highlight the effect of the beamwidth on the position information gain. In particular, we choose $N_{\mathrm{TX}}=25$ (Fig. \ref{fig:pos_info_heatmap} - left) and $N_{\mathrm{TX}}=150$ (Fig. \ref{fig:pos_info_heatmap} - right). Wider beams of the former array result in a homogeneous illumination of the plane, which makes it more obvious to point out the location dependency of the point of incidence on the net position information gain. Three main conclusions can be drawn: 
\begin{enumerate}
	\item First, the geometry of the scenario has a significant impact on the position information gain. The results in Fig. \ref{fig:pos_info_heatmap} (left) confirm our findings from the analysis on the position information gain in section \ref{subsec:eff_MP_gain}. In particular, points of incidence that are close to the transmitter and receiver provide large information gains. In addition, certain angles $\Delta\theta_k$ invoke larger net position information gain than others. This can be deduced from the inhomogeneous color pattern in Fig. \ref{fig:pos_info_heatmap} (left). 
	\item  Secondly, the illumination of the plane with the transmitted beams has a major impact on the position information gain of NLOS components. Generally, narrower beams generated by larger apertures result in larger net position information gains. Thus, the points of incidence of NLOS paths should be illuminated with beams, which are as narrow as possible, in order to obtain the highest increase in the positioning accuracy.
	\item The fact that NLOS components provide information regarding the position of mobile terminal implies that the points of incidence of NLOS paths themselves can be estimated. Hence simultaneous localization and mapping (SLAM) can be conducted within a single snapshot. Thus, if SLAM is the goal on the system level, densely-spaced narrow beams are preferable over wide beams in order to accurately estimate the positions of the reflectors and achieve high localization accuracy of the mobile terminal. A synergy between communication and environmental mapping can be identified. When highly accurate maps of the environment have been obtained, e.g. in a discovery phase, they can be stored in a cloud database. Such a cloud database can, in turn, be used to support communication, since beams could intentionally be steered towards points of incidence to increase the capacity or reliably of the link. 
\end{enumerate} 
\begin{figure*}[t]%
\centering
\hspace*{\fill}%
\begin{subfigure}{.48\columnwidth}
\centering
\includegraphics[width=\columnwidth]{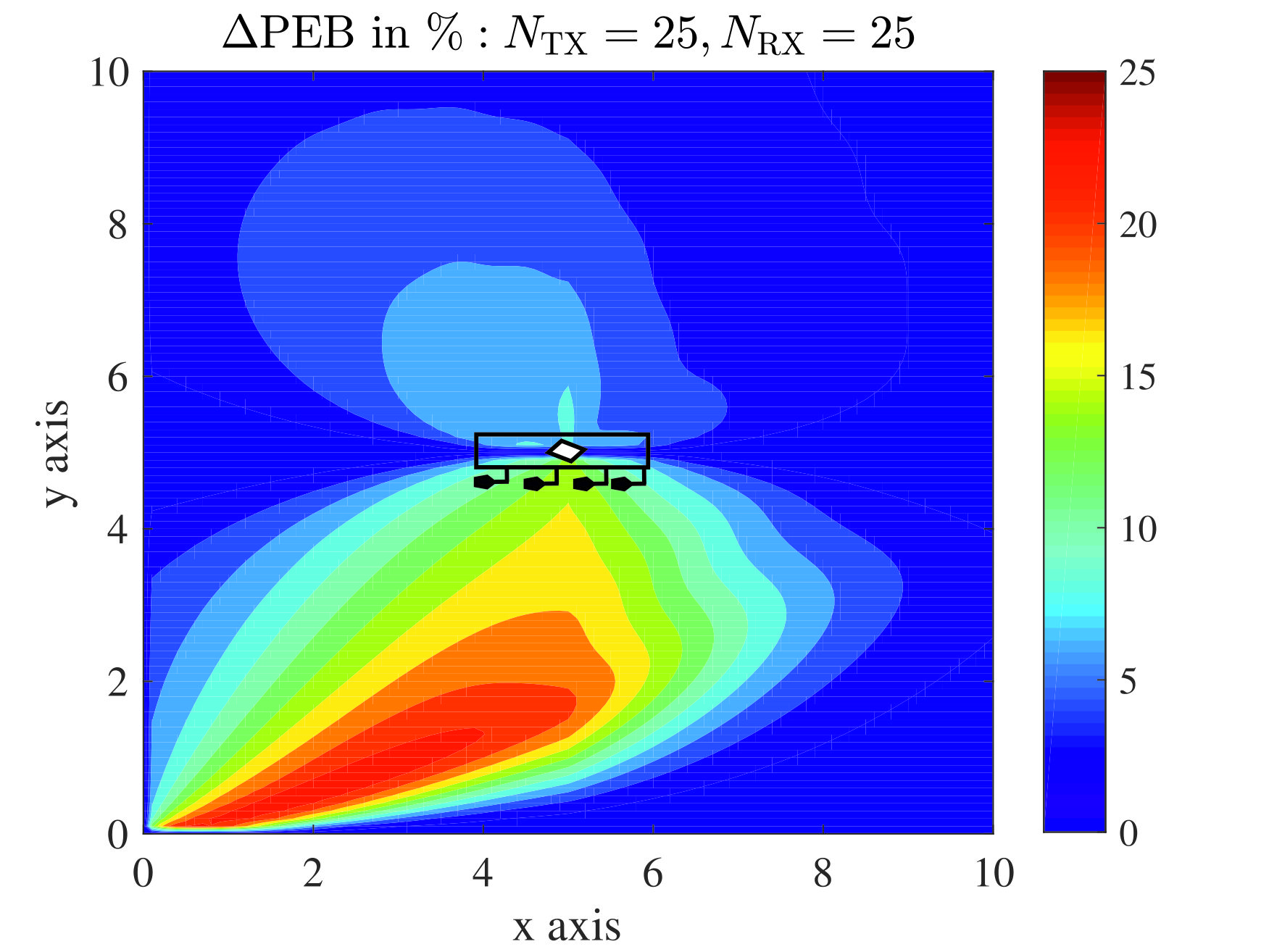}%
\label{fig:PEB_NTX_25_NRX_150_NB_50}%
\end{subfigure}
\begin{subfigure}{.48\columnwidth}
\centering
\includegraphics[width=\columnwidth]{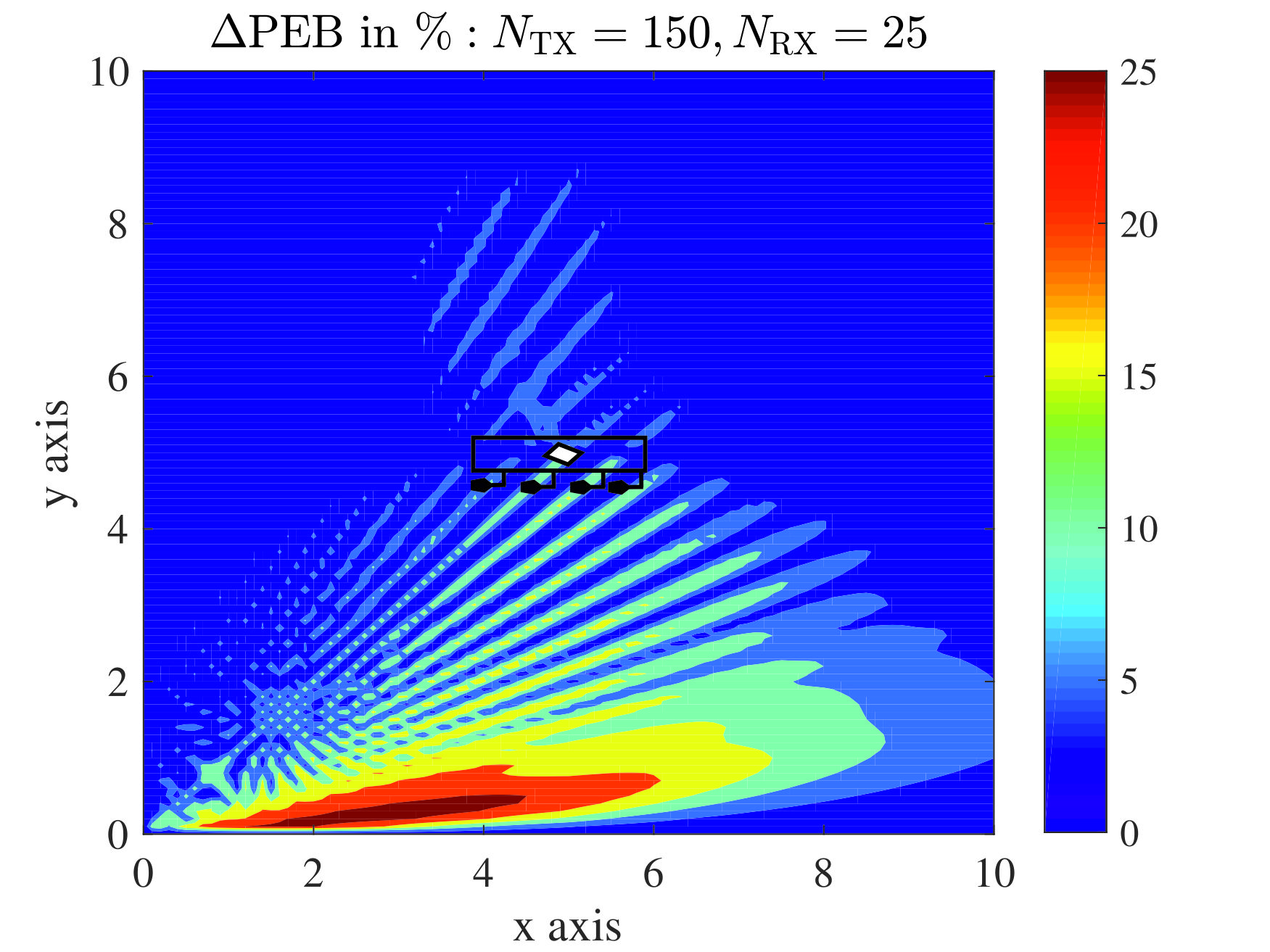}%
\label{fig:PEB_NTX_150_NRX_150_NB_50}%
\end{subfigure}
\caption{\textit{Reduction of the PEB -} The patterns of the reductions of the PEBs closely resemble the patterns of the net position information gains in Fig. \ref{fig:pos_info_heatmap}.}
\label{fig:PEB_reduction}%
\end{figure*}
Even though, we did not analytically show the reduction of the PEB in the presence of a NLOS path, we provide a numerical analysis. To that end, we determine the reduction of the position error bound $\Delta \mathrm{PEB}$ in the presence of a single NLOS paths and depict a heat map in Fig. \ref{fig:PEB_reduction}. We observe that the presence of the reflector reduces the PEB by up to $25\%$. Narrower beams (Fig. \ref{fig:PEB_reduction} - right) result in larger reductions of the PEB. Note that the patterns in Fig. \ref{fig:PEB_reduction} closely resemble the patterns of the net position information gain in Fig. \ref{fig:pos_info_heatmap}, i.e. when the net position information gain is large, the reduction of the PEB is also large. 
\subsection{Example 2}
We assume that the transmitter and the receiver are equipped with $N_{\mathrm{TX}}=N_{\mathrm{TX}}=25$, and the plane is illuminated with $N_{\mathrm{B}}=50$ beams. We assume two different conditions regarding the received paths: 1) the LOS path is present and 2) the LOS path is obstructed. In the former scenario, there are no NLOS paths, while in the latter scenario we consider three NLOS paths. Reflectors which cause these NLOS paths are located at $\vec{s}_1=[8, 1]^{\mathrm{T}}$, $\vec{s}_2=[3, 4]^{\mathrm{T}}$, and  $\vec{s}_3=[6, 8]^{\mathrm{T}}$. The left and right part of Fig. \ref{fig:p_gain_n_direct} depict the first and second scenario, respectively. The direction of information (eigenvector) is indicated by normalized colored arrows, while the (net) position information gain (eigenvalue) is presented in text boxes attached to the arrows. The corresponding PEB is depicted in the title of the respective sub-figure. We can deduce from Fig. \ref{fig:p_gain_n_direct} (left) that unambiguous single-anchor position and orientation estimation is possible. Fig. \ref{fig:p_gain_n_direct} (right) demonstrates that, even in the absence of LOS, three NLOS paths lead to a full rank EFIM in \eqref{eq:EFIM_corollary}, which reflects unambiguous position and orientation estimation.  Note that there is a degradation of the PEB of approximately one decade when the LOS path is replaced by 3 NLOS paths. Nonetheless, the position of the mobile terminal can determined accurately based on NLOS paths only. 
\begin{figure*}[t]%
\centering
\hspace*{\fill}%
\begin{subfigure}{.48\columnwidth}
\centering
\includegraphics[width=\columnwidth]{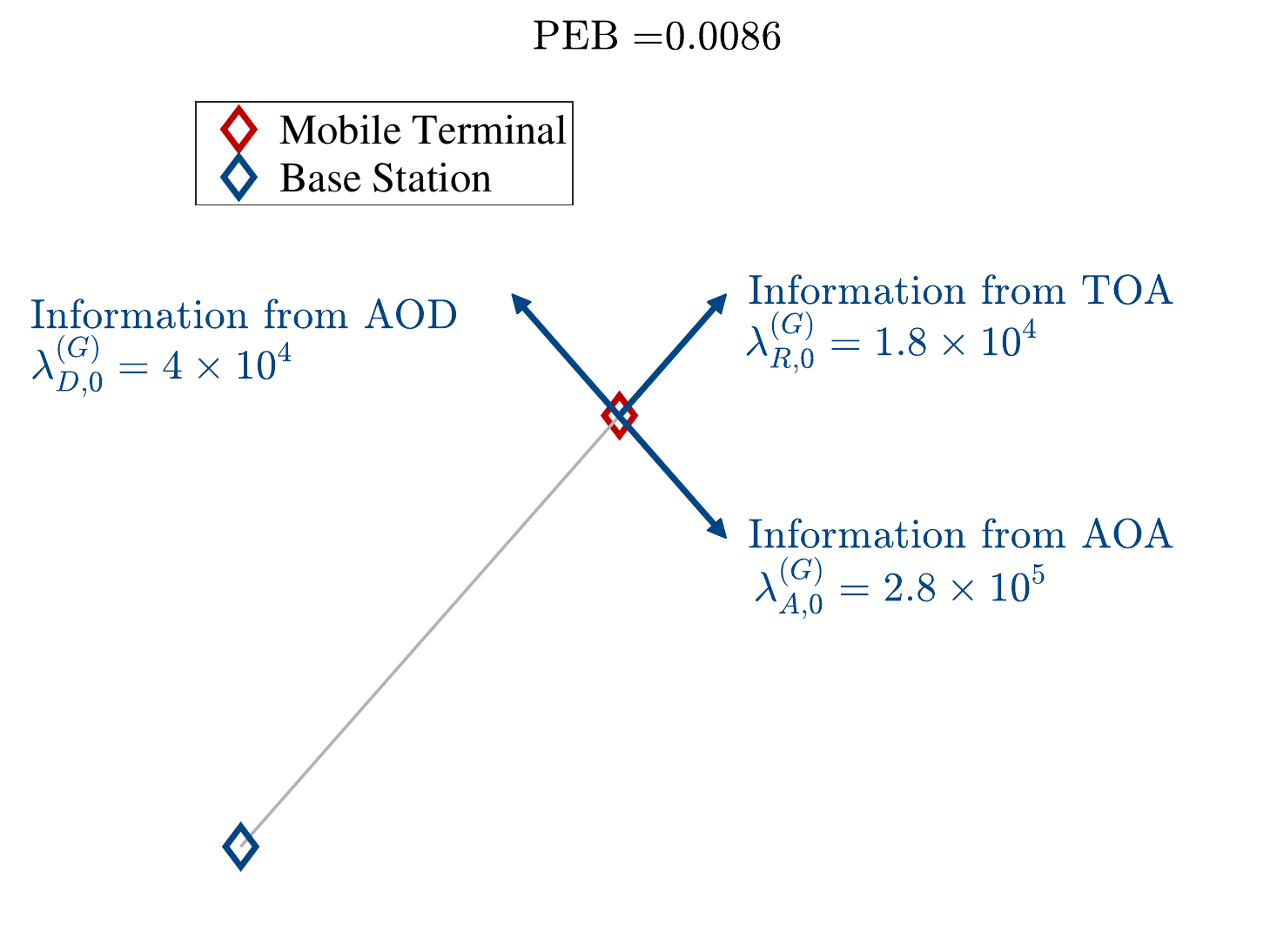}%
\label{fig:LOS_p_gain_n_direct}%
\end{subfigure}
\begin{subfigure}{.48\columnwidth}
\centering
\includegraphics[width=\columnwidth]{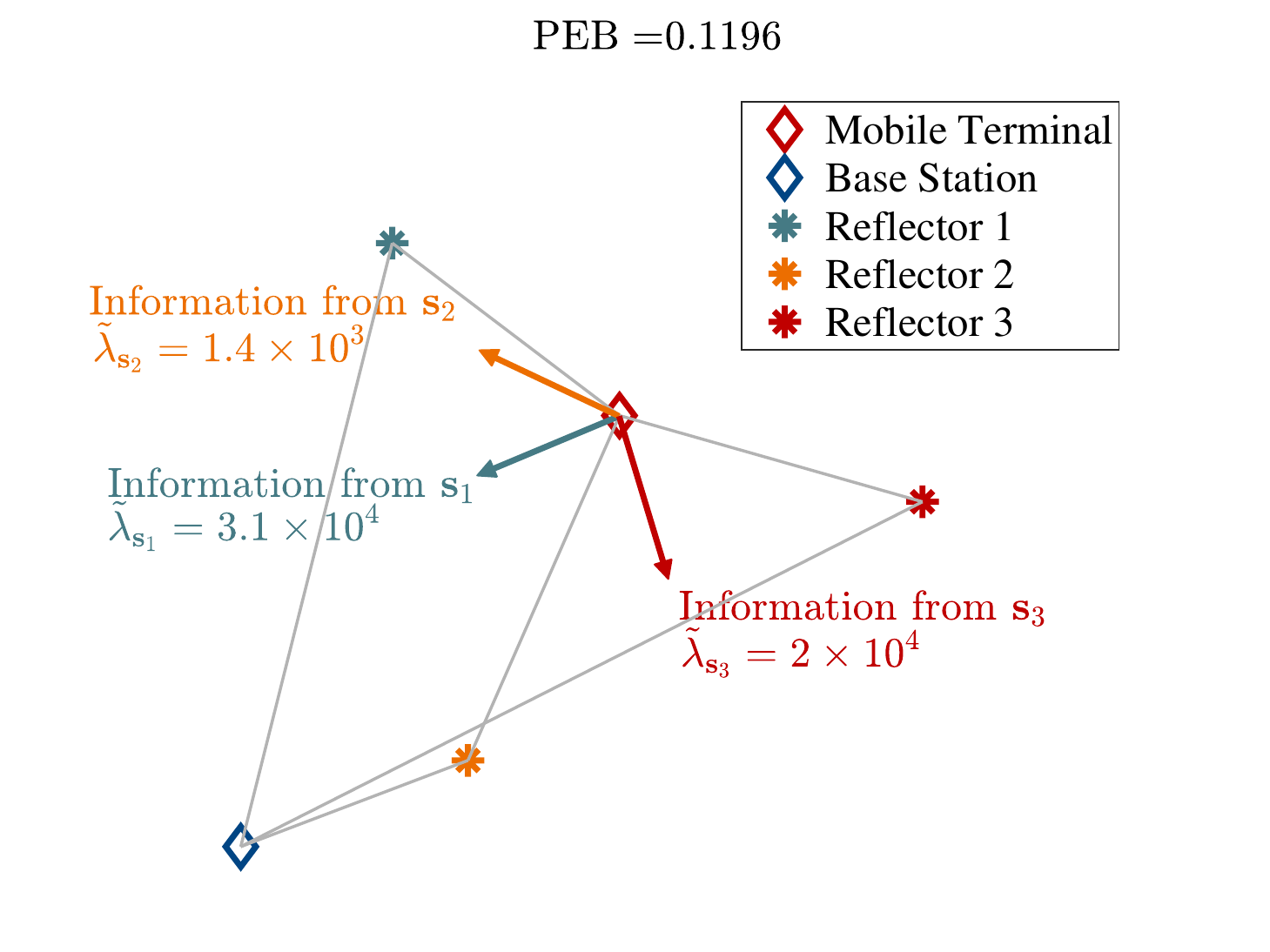}%
\label{fig:NLOS_p_gain_n_direct}%
\end{subfigure}
\caption{Left: position information gains and directions originating from TOA, AOD, and AOA of the LOS path. Right: net position information gains and directions of 3 NLOS paths. Accurate positioning is possible even in the absence of the LOS path.}
\label{fig:p_gain_n_direct}%
\end{figure*}
\section{Conclusion}
\label{sec:conclusion}
We analyzed the role of NLOS components in 5G mmWave MIMO systems in terms of their position and orientation estimation capabilities. For our analysis, we employed the concept of Fisher information in order to show that NLOS components provide significant Fisher information if and only if angle-of-arrival, angle-of-departure, and time-of-arrival of the corresponding path can be estimated accurately. We showed analytically that each NLOS component contributes one dimensional Fisher information regarding the position and orientation. Hence NLOS components can be harnessed to increase position and orientation accuracy. We showed that even in the absence of the LOS component, unambiguous position and orientation estimation is feasible if at least three NLOS paths contribute to the received signal. We showed in our analysis that the amount of gained information strongly depends on the relative position of the mobile terminal, the base station, and the points of incidence, as well as the illumination of the point of incidence by base station. We pointed out that narrow beams increase the position and orientation information gain of NLOS components when compared to wider beams.
\appendices
\section{Entries of the Transformation Matrix}
\label{sec:appendix_A}
The partial derivates in the matrix $\vec{T}$ are given by \cite{SGDSW2017}
\begin{subequations}
\begin{eqnarray}
\frac{\partial\tau_0}{\partial\vec{p}} &=& \frac{1}{c}[\cos(\theta_{\mathrm{TX},0}),\sin(\theta_{\mathrm{TX},0})]^{\mathrm{T}},
\label{eq:dtau_0_dp}\\
\frac{\partial\theta_{\mathrm{TX},0}}{\partial\vec{p}} &=& \frac{1}{\left\|\vec{p}-\vec{q}\right\|_2}[-\sin(\theta_{\mathrm{TX},0}),\cos(\theta_{\mathrm{TX},0})]^{\mathrm{T}},
\label{eq:dtheta_TX0_dp}\\
\frac{\partial\theta_{\mathrm{RX},0}}{\partial\vec{p}} &=&\frac{1}{\left\|\vec{p}-\vec{q}\right\|_2}[-\sin(\theta_{\mathrm{TX},0}),\cos(\theta_{\mathrm{TX},0})]^{\mathrm{T}},
\label{eq:dtheta_RX0_dp}\\
\frac{ \partial\tau_k}{\partial\vec{p}} &=& \frac{1}{c}[\cos(\pi-\theta_{\mathrm{RX},k}),-\sin(\pi-\theta_{\mathrm{RX},k})]^{\mathrm{T}}, k>0,
\label{eq:dtau_k_dp}\\
\frac{\partial\theta_{\mathrm{RX},k}}{\partial\vec{p}} &=& \frac{1}{\left\|\vec{p}-\vec{s}_k\right\|_2}[\sin(\pi-\theta_{\mathrm{RX},k}),\cos(\pi-\theta_{\mathrm{RX},k})]^{\mathrm{T}}, k>0,
\label{eq:dtheta_RXk_dp}\\
\frac{\partial\tau_k}{\partial\vec{s}_k} &=& \frac{1}{c}[\cos(\theta_{\mathrm{TX},k})+\cos(\theta_{\mathrm{RX},k}),\sin(\theta_{\mathrm{TX},k})+\sin(\theta_{\mathrm{RX},k})]^{\mathrm{T}}, k>0,
\label{eq:dtau_k_dsk}\\
\frac{\partial\theta_{\mathrm{TX},k}}{\partial\vec{s}_k} &=& \frac{1}{\left\|\vec{q}-\vec{s}_k\right\|_2}[-\sin(\theta_{\mathrm{TX},k}),\cos(\theta_{\mathrm{TX},k})]^{\mathrm{T}}, k>0,
\label{eq:dtheta_TXk_dsk}\\
\frac{\partial\theta_{\mathrm{RX},k}}{\partial\vec{s}_k} &=& -\frac{1}{\left\|\vec{p}-\vec{s}_k\right\|_2}[\sin(\pi-\theta_{\mathrm{RX},k}),\cos(\pi-\theta_{\mathrm{RX},k})]^{\mathrm{T}}, k>0,
\label{eq:dtheta_RXk_dsk}\\
\frac{\partial\theta_{\mathrm{RX},k}}{\partial\alpha} &=& -1, k\geq 0.
\label{eq:dtheta_RXk_dalpha}
\end{eqnarray}
\end{subequations}
All other partial derivates are zero. Considering \eqref{eq:dtau_0_dp}-\eqref{eq:dtheta_RXk_dalpha}, we find 
\begin{equation}
\vec{T}_{\vec{P}}^{(0)} = 
\left[
\begin{array}{ccc}
	\frac{1}{c}\cos(\theta_{\mathrm{TX},0}) & -\frac{1}{\left\|\vec{p}-\vec{q}\right\|_2}\sin(\theta_{\mathrm{TX},0}) & -\frac{1}{\left\|\vec{p}-\vec{q}\right\|_2}\sin(\theta_{\mathrm{TX},0}) 
	\\
	\frac{1}{c}\sin(\theta_{\mathrm{TX},0}) & \frac{1}{\left\|\vec{p}-\vec{q}\right\|_2}\cos(\theta_{\mathrm{TX},0}) &\frac{1}{\left\|\vec{p}-\vec{q}\right\|_2}\cos(\theta_{\mathrm{TX},0}) 
	\\
	0 												& 0 																																& -1
\end{array}
\right].
\label{eq:mat_T_P_0}
\end{equation}
In addition, we find 
\begin{equation}
\vec{T}_{\vec{P}}^{(k)} = 
\left[
\begin{array}{ccc}
	\frac{1}{c}\cos(\pi-\theta_{\mathrm{RX},k}) & 0 & -\frac{1}{\left\|\vec{p}-\vec{s}_k\right\|_2}\sin(\pi-\theta_{\mathrm{RX},k}) 
	\\
	-\frac{1}{c}\sin(\pi-\theta_{\mathrm{RX},k}) & 0 &\frac{1}{\left\|\vec{p}-\vec{s}_k\right\|_2}\cos(\pi-\theta_{\mathrm{RX},k}) 
	\\
	0 															& 0 & -1
\end{array}
\right]
\label{eq:mat_T_P_k}
\end{equation}
and, finally, 
\begin{equation}
\vec{T}_{\vec{s}_k} = 
\left[
\begin{array}{ccc}
	\frac{1}{c}\left[\cos(\theta_{\mathrm{TX},k})+\cos(\theta_{\mathrm{RX},k}) \right]& -\frac{1}{\left\|\vec{q}-\vec{s}_k\right\|_2}\sin(\theta_{\mathrm{TX},k}) & -\frac{1}{\left\|\vec{p}-\vec{s}_k\right\|_2}\sin(\pi-\theta_{\mathrm{RX},k})
	\\
	\frac{1}{c}\left[\sin(\theta_{\mathrm{TX},k})+\sin(\theta_{\mathrm{RX},k}) \right]& \frac{1}{\left\|\vec{q}-\vec{s}_k\right\|_2}\cos(\theta_{\mathrm{TX},k})  &-\frac{1}{\left\|\vec{p}-\vec{s}_k\right\|_2}\cos(\pi-\theta_{\mathrm{RX},k})
\end{array}
\right],
\label{eq:mat_T_s_k}
\end{equation}
for $k>0$.
\section{Eigenvalues and Eigenvectors of the LOS Information Gain Matrices}
\label{sec:appendix_B}
Since the matrices in \eqref{eq:EFIM_gain_LOS} are rank one, the eigenvalues are given by the traces of the respective matrices \cite{L2004}. Hence \eqref{eq:LOS_eval_R}-\eqref{eq:LOS_eval_A} follow.  

Having obtained the eigenvalues of the matrices $\tilde{\vec{A}}_R^{\mathrm{(G)}}$,$\tilde{\vec{A}}_D^{\mathrm{(G)}}$, and $\tilde{\vec{A}}_A^{\mathrm{(G)}}$, it is straightforward to see that the vectors $\vec{v}_{R,0}^{\mathrm{(G)}}$,$\vec{v}_{D,0}^{\mathrm{(G)}}$, and $\vec{v}_{A,0}^{\mathrm{(G)}}$ in \eqref{eq:LOS_evec_R}-\eqref{eq:LOS_evec_A} are in null space of $(\lambda_R^{\mathrm{(G)}}\vec{I} -\tilde{\vec{A}}_R^{\mathrm{(G)}})$, $(\lambda_D^{\mathrm{(G)}}\vec{I}- \tilde{\vec{A}}_D^{\mathrm{(G)}})$, and $(\lambda_A^{\mathrm{(G)}}\vec{I} -\tilde{\vec{A}}_A^{\mathrm{(G)}})$, respectively, i.e. 
\begin{subequations}
\begin{alignat}{3}
\left(\lambda_R^{\mathrm{(G)}}\vec{I} -\tilde{\vec{A}}_R^{\mathrm{(G)}}\right)\vec{v}_{R,0}^{\mathrm{(G)}}&=\vec{0},
\\ 
\left(\lambda_D^{\mathrm{(G)}}\vec{I}- \tilde{\vec{A}}_D^{\mathrm{(G)}}\right)\vec{v}_{D,0}^{\mathrm{(G)}}&=\vec{0},
\\
\left(\lambda_A^{\mathrm{(G)}}\vec{I}- \tilde{\vec{A}}_A^{\mathrm{(G)}}\right)\vec{v}_{A,0}^{\mathrm{(G)}}&=\vec{0}.
\end{alignat}
\end{subequations}
\section{Decomposition of the NLOS Information Loss}
\label{sec:appendix_C}
In the following, we will derive the decomposition of the terms $\tilde{\vec{B}}^{\mathrm{(G)}}$, $\tilde{\vec{B}}^{\mathrm{(L)}}$, and 
\newline $\tilde{\vec{B}}^{\mathrm{(N)}}=\tilde{\vec{B}}^{\mathrm{(G)}}-\tilde{\vec{B}}^{\mathrm{(L)}}$.
\subsection{NLOS Information Gain}
\label{sec:appendix_C_MP_gain}
The NLOS information gain is given by 
\begin{equation}
\begin{array}{rl}
\tilde{\vec{B}}^{\mathrm{(G)}}&\triangleq  \mathbf{B}[\mathbf{J}_{{\boldsymbol\eta}}]_{4:3K,4:3K}\mathbf{B}^{\mathrm{T}} \\
& = 
[\vec{T}_{\vec{P}}^{(1)}, \vec{T}_{\vec{P}}^{(2)},...,\vec{T}_{\vec{P}}^{(K-1)}] [\mathbf{J}_{{\boldsymbol\eta}}]_{4:3K,4:3K}[\vec{T}_{\vec{P}}^{(1)}, \vec{T}_{\vec{P}}^{(2)},...,\vec{T}_{\vec{P}}^{(K-1)}]^{\mathrm{T}}.	
\end{array}
\label{eq:mp_gain_derv1}
\end{equation}
Since 
\begin{equation}
[\mathbf{J}_{{\boldsymbol\eta}}]_{4:3K,4:3K}=\text{diag}[1/\sigma^2_{\tau_1},1/\sigma^2_{\theta_{\mathrm{TX},1}},1/\sigma^2_{\theta_{\mathrm{RX},1}},...,1/\sigma^2_{\tau_{K-1}},1/\sigma^2_{\theta_{\mathrm{TX},K-1}},1/\sigma^2_{\theta_{\mathrm{RX},K-1}}],
\label{eq:J_eta_mp}
\end{equation}
we find
\begin{equation}
\begin{array}{cl}
	\tilde{\vec{B}}^{\mathrm{(G)}}&= \sum_{k=1}^{K-1} \vec{T}_{\vec{P}}^{(k)}\vec{J}_{\bar{\boldsymbol\eta}_k} \left(\vec{T}_{\vec{P}}^{(k)}\right)^{\mathrm{T}} \\
&= \sum_{k=1}^{K-1}\frac{1}{\sigma^2_{\tau_k}c^2} \boldsymbol\Upsilon_{0,0}(\theta_{\mathrm{RX},k},0)
	+\frac{1}{\sigma^2_{\theta_{\mathrm{RX},k}}\left\|\vec{p}-\vec{s}_k\right\|^2} \boldsymbol\Upsilon_{1,1}(\theta_{\mathrm{RX},k},\left\|\vec{p}-\vec{q}\right\|). 
\end{array}
\label{eq:mp_gain_derv2}
\end{equation}

\subsection{NLOS Information Loss}
\label{sec:appendix_C_MP_loss}
It is easy to verify that $\vec{D}[\vec{J}_{{\bar{\boldsymbol\eta}}}]_{4:3K,4:3K}\vec{D}^{\mathrm{T}}$ is block diagonal since $\vec{D}$ is block diagonal and from \eqref{eq:J_eta_mp} we recall that $[\vec{J}_{{\bar{\boldsymbol\eta}}}]_{4:3K,4:3K}$ is diagonal. Hence 
\begin{equation}
\mathbf{D}[\mathbf{J}_{{\boldsymbol\eta}}]_{4:3K,4:3K}\mathbf{D}^{\mathrm{T}}=
\begin{bmatrix}
\mathbf{T}_{\mathbf{s}_1}\vec{J}_{\bar{\boldsymbol\eta}_1}\mathbf{T}_{\mathbf{s}_1}^{\mathrm{T}}
&\hdots &
\mathbf{0}_{2\times2}
\\
\vdots&
\ddots&
\vdots
\\
\mathbf{0}_{2\times2} &
\hdots &
\mathbf{T}_{\mathbf{s}_{K-1}}\vec{J}_{\bar{\boldsymbol\eta}_{K-1}}\mathbf{T}_{\mathbf{s}_{K-1}}^{\mathrm{T}}
\\
\end{bmatrix}.
\label{eq:middle_term_mp_loss_derv1}
\end{equation}
For the compactness of notation, we use the following shorthand 
\newline $\tilde{\vec{T}}_{\vec{s}_k\vec{s}_k}\triangleq \mathbf{T}_{\mathbf{s}_{k}}\vec{J}_{\bar{\boldsymbol\eta}_{k}}\mathbf{T}_{\mathbf{s}_{k}}^{\mathrm{T}}$ The inverse of a block diagonal matrix is determined by inverting the blocks on the diagonal. The $k^{\text{th}}$ block is given by
\begin{equation}
\tilde{\vec{T}}_{\vec{s}_k\vec{s}_k} \triangleq 
\left[
\begin{array}{cc}
 a_k & b_k  \\
 b_k & d_k	
\end{array}
\right],
\label{eq:middle_term_mp_loss_derv2}
\end{equation}
where 
\begin{align*}
a_k \triangleq& \frac{(\cos(\theta_{\mathrm{TX},k})+\cos(\theta_{\mathrm{RX},k}))^2}{\sigma^2_{\tau_k}c^2}+\frac{\sin^2(\theta_{\mathrm{TX},k})}{\sigma^2_{\theta_{\mathrm{TX},k}}\left\|\mathbf{q} - \mathbf{s_k}\right\|^2} +\frac{\sin^2(\theta_{\mathrm{RX},k})}{\sigma^2_{\theta_{\mathrm{RX},k}}\left\|\mathbf{p} - \mathbf{s_k}\right\|^2}
\\
b_k \triangleq &\frac{(\cos(\theta_{\mathrm{TX},k})+\cos(\theta_{\mathrm{RX},k}))(\sin(\theta_{\mathrm{TX},k})+\sin(\theta_{\mathrm{RX},k}))}{\sigma^2_{\tau_k}c^2}
\\
& -\frac{\sin(\theta_{\mathrm{TX},k})\cos(\theta_{\mathrm{TX},k})}{\sigma^2_{\theta_{\mathrm{TX},k}}\left\|\mathbf{q} - \mathbf{s_k}\right\|^2}-\frac{\sin(\theta_{\mathrm{RX},k})\cos(\theta_{\mathrm{RX},k})}{\sigma^2_{\theta_{\mathrm{RX},k}}\left\|\mathbf{p} - \mathbf{s_k}\right\|^2} 
\\
d_k \triangleq &\frac{(\sin(\theta_{\mathrm{TX},k})+\sin(\theta_{\mathrm{RX},k}))^2}{\sigma^2_{\tau_k}c^2}+\frac{\cos^2(\theta_{\mathrm{TX},k})}{\sigma^2_{\theta_{\mathrm{TX},k}}\left\|\mathbf{q} - \mathbf{s_k}\right\|^2} +\frac{\cos^2(\theta_{\mathrm{RX},k})}{\sigma^2_{\theta_{\mathrm{RX},k}}\left\|\mathbf{p} - \mathbf{s_k}\right\|^2}.
\end{align*}
Thus, the inverse of the $k^{\text{th}}$ block is given by
\begin{equation}
\tilde{\vec{T}}_{\vec{s}_k\vec{s}_k}^{-1} 
=
\frac{1}{\left|\tilde{\vec{T}}_{\vec{s}_k\vec{s}_k}\right|}
\left[
\begin{array}{cc}
 d_k & -b_k  \\
 -b_k & a_k	
\end{array}
\right],
\label{eq:middle_term_mp_loss_derv3}
\end{equation}
where $\left|\tilde{\vec{T}}_{\vec{s}_k\vec{s}_k} \right|=a_k d_k - b_k^2$.

The left term $\vec{B}[\vec{J}_{{\bar{\boldsymbol\eta}}}]_{4:3K,4:3K}\vec{D}^{\mathrm{T}}$ and the right term $\vec{D}[\vec{J}_{{\bar{\boldsymbol\eta}}}]_{4:3K,4:3K}\vec{B}^{\mathrm{T}}$ in \eqref{eq:MP_loss_as_sum} can be evaluated by straightforward matrix-matrix multiplications. Observe that the left term can be obtained by taking the transpose of the right term. Due to the block diagonal structure of $\vec{D}$, we find 
\begin{equation}
\mathbf{D}[\mathbf{J}_{{\bar{\boldsymbol\eta}}}]_{4:3K,4:3K}\mathbf{B}^{\mathrm{T}}=
\begin{bmatrix}
\mathbf{T}_{\mathbf{s}_1}\vec{J}_{\bar{\boldsymbol\eta}_{1}}\left(\mathbf{T}_{\mathbf{P}}^{(1)}\right)^{\mathrm{T}} &
\\
\vdots
\\
\mathbf{T}_{\mathbf{s}_{K-1}}\vec{J}_{\bar{\boldsymbol\eta}_{K-1}}\left(\mathbf{T}_{\mathbf{P}}^{(K-1)}\right)^{\mathrm{T}}
\\
\end{bmatrix}.
\label{eq:right_term_mp_loss_derv1}
\end{equation}
We introduce a similar shorthand as above, where $\tilde{\vec{T}}_{\vec{s}_k\vec{P}}\triangleq \mathbf{T}_{\mathbf{s}_{k}}\vec{J}_{\bar{\boldsymbol\eta}_{k}}\left(\mathbf{T}_{\mathbf{P}}^{(k)}\right)^{\mathrm{T}}$. Considering \eqref{eq:middle_term_mp_loss_derv3} and \eqref{eq:right_term_mp_loss_derv1}, we deduce
\begin{equation}
\tilde{\vec{B}}^{\mathrm{(L)}}
= 
\sum_{k=1}^{K-1}    \tilde{\vec{T}}_{\vec{s}_k\vec{P}}^{\mathrm{T}}  \tilde{\vec{T}}_{\vec{s}_k\vec{s}_k}^{-1}\tilde{\vec{T}}_{\vec{s}_k\vec{P}}.
\label{eq:right_term_mp_loss_derv2}
\end{equation}
Let us define the following weights which are used in \eqref{eq:EFIM_loss_MP}:\newline
\begin{subequations}
\begin{alignat}{3}
w_{R,k}^{\mathrm{(L)}}&\triangleq \left(\frac{1+\cos(\Delta\theta_k)}{\sigma_{\tau_k}^2c^2}\right)^2\left(\frac{1}{\sigma^2_{\theta_{\mathrm{RX},k}}\left\|\mathbf{p} - \mathbf{s_k}\right\|^2}+\frac{1}{\sigma^2_{\theta_{\mathrm{TX},k}}\left\|\mathbf{q} - \mathbf{s_k}\right\|^2}\right)\frac{1}{a_k d_k - b_k^2}, 
\\
w_{A,k}^{\mathrm{(L)}}&\triangleq \left(\frac{\left(1+\cos(\Delta\theta_k)\right)^2}{\sigma_{\tau_k}^2c^2\sigma^4_{\theta_{\mathrm{RX},k}}\left\|\mathbf{p} - \mathbf{s_k}\right\|^4} +\frac{\sin^2(\Delta\theta_k)}{\sigma^2_{\theta_{\mathrm{TX},k}}\left\|\mathbf{q} - \mathbf{s_k}\right\|^2\sigma^4_{\theta_{\mathrm{RX},k}}\left\|\mathbf{p} - \mathbf{s_k}\right\|^4}\right)\frac{1}{a_k d_k - b_k^2},
\\
\gamma_{\vec{s}_k}&\triangleq \frac{(1+\cos(\Delta\theta_k))\sin(\Delta\theta_k)}{\sigma^2_{\theta_{\mathrm{TX},k}}\left\|\mathbf{q} - \mathbf{s_k}\right\|^2\sigma^2_{\theta_{\mathrm{RX}, k}}\left\|\mathbf{p} - \mathbf{s_k}\right\|^2\sigma^2_{\tau_k} c^2}\frac{1}{a_k d_k - b_k^2}.
\end{alignat}
\end{subequations}
In addition, we define 
\begin{equation}
	\vec{B}_{k}^{\mathrm{(L)}}
	\triangleq
	\begin{bmatrix}
-2\sin(\theta_{\mathrm{RX},k})\cos(\theta_{\mathrm{RX},k})& 
\cos^2(\theta_{\mathrm{RX},k})-\sin^2(\theta_{\mathrm{RX},k})& 
\cos(\theta_{\mathrm{RX},k})\left\|\mathbf{p}-\mathbf{s}_k\right\| \\
\cos^2(\theta_{\mathrm{RX},k})-\sin^2(\theta_{\mathrm{RX},k}) & 
2\sin(\theta_{\mathrm{RX},k})\cos(\theta_{\mathrm{RX},k}) &
\sin(\theta_{\mathrm{RX},k})\left\|\mathbf{p}-\mathbf{s}_k\right\|\\
\cos(\theta_{\mathrm{RX},k})\left\|\mathbf{p}-\mathbf{s}_k\right\| & 
\sin(\theta_{\mathrm{RX},k})\left\|\mathbf{p}-\mathbf{s}_k\right\| &
0\\
\end{bmatrix}.
\label{eq:MP_eff_loss_RAD}
\end{equation}
Using straightforward matrx-matrix multiplication, \eqref{eq:EFIM_loss_MP} is readily obtained.
\subsection{Net NLOS Gain}
\label{sec:appendix_C_MP_eff_gain}
We collect all information associated with the $k^{\text{th}}$ NLOS component in the matrix $\boldsymbol\Psi_{\mathbf{s}_k}$, i.e.
\begin{equation}
\begin{array}{cl}
\boldsymbol\Psi_{\mathbf{s}_k} 
&\triangleq
\epsilon_{\mathbf{s}_k} \boldsymbol\Upsilon_{0,0}(\theta_{\mathrm{TX},0},0)+\beta_{\mathbf{s}_k}\boldsymbol\Upsilon_{1,1}(\theta_{\mathrm{RX},k},\left\|\vec{p}-\vec{q}\right\|)
	+ \gamma_{\mathbf{s}_k} \vec{B}_{k}^{\mathrm{(G)}}
	\\
	&\triangleq
 [\boldsymbol\psi_{\mathbf{s}_k,1},\boldsymbol\psi_{\mathbf{s}_k,2},\boldsymbol\psi_{\mathbf{s}_k,3}].
\label{eq:psi_sk}
\end{array}
\end{equation}
We will now show that $\boldsymbol\Psi_{\mathbf{s}_k}$ has only one non-zero eigenvalue by showing that $\boldsymbol\Psi_{\mathbf{s}_k}$ is rank one, i.e. $\boldsymbol\Psi_{\mathbf{s}_k}$ has only one linearly independent column. By inspection of the columns of $\boldsymbol\Psi_{\mathbf{s}_k}$, it is straightforward to show that  
\begin{equation}
a\boldsymbol\psi_{\mathbf{s}_k,1}
-
b\boldsymbol\psi_{\mathbf{s}_k,2}
=
\begin{bmatrix}
\cos(\theta_{\mathrm{RX},k})\left\|\mathbf{p}-\mathbf{s}_k\right\|(\epsilon_{\mathbf{s}_k}\beta_{\mathbf{s}_k}-\gamma_{\mathbf{s}_k}^2)
\\
\sin(\theta_{\mathrm{RX},k})\left\|\mathbf{p}-\mathbf{s}_k\right\|(\epsilon_{\mathbf{s}_k}\beta_{\mathbf{s}_k}-\gamma_{\mathbf{s}_k}^2)
\\
0
\end{bmatrix},
\label{eq:lin_dep_12_detailed}
\end{equation}
where $a\triangleq(\beta_{\mathbf{s}_k}\cos(\theta_{\mathrm{RX},k})\left\|\mathbf{p}-\mathbf{s}_k\right\|+\gamma_{\mathbf{s}_k}\sin(\theta_{\mathrm{RX},k})\left\|\mathbf{p}-\mathbf{s}_k\right\|)$ and 
\newline $b\triangleq(-\beta_{\mathbf{s}_k}\sin(\theta_{\mathrm{RX},k})\left\|\mathbf{p}-\mathbf{s}_k\right\|+\gamma_{\mathbf{s}_k}\cos(\theta_{\mathrm{RX},k})\left\|\mathbf{p}-\mathbf{s}_k\right\|)$. Using simple algebra it can be seen that
\begin{equation}
(\epsilon_{\mathbf{s}_k}\beta_{\mathbf{s}_k}-\gamma_{\mathbf{s}_k}^2)=0.
\label{eq:eps_beta_gamma}
\end{equation} 
Hence $\boldsymbol\psi_{\mathbf{s}_k,1}$ and $\boldsymbol\psi_{\mathbf{s}_k,2}$ are linearly dependent. All other combinations of the columns of $\boldsymbol\Psi_{\mathbf{s}_k}$ can be shown to linearly dependent in the same fashion. Hence $\boldsymbol\Psi_{\mathbf{s}_k}$ is rank one and has one non-zero eigenvalue. The only non-zero eigenvalue is given by the trace of $\boldsymbol\Psi_{\mathbf{s}_k}$ \cite{L2004}
\begin{equation}
\lambda_{\mathbf{s}_k}=\epsilon_{\mathbf{s}_k}+\beta_{\mathbf{s}_k}\left(1+\left\|\vec{p}-\vec{s}_k\right\|\right).
\label{eq:non_zero_eval_proof}
\end{equation}
Using the results for $\epsilon_{\vec{s}_k}$, and $\beta_{\vec{s}_k}$, as well as and straightforward algebraic manipulations, the main result of this paper in \eqref{eq:non_zero_eigenvalue} is readily obtained.
That being said, it can be seen that 
\begin{equation}
\mathbf{v}_{\mathbf{s}_k}=
\begin{bmatrix}
-\frac{1}{\left\|\mathbf{p}-\mathbf{s}_k\right\|}\left(\frac{\epsilon_{\mathbf{s}_k}}{\gamma_{\mathbf{s}_k}}\cos(\theta_{\mathrm{RX},k})+ \sin(\theta_{\mathrm{RX},k})\right) &
\\
\frac{1}{\left\|\mathbf{p}-\mathbf{s}_k\right\|}\left(-\frac{\epsilon_{\mathbf{s}_k}}{\gamma_{\mathbf{s}_k}}\sin(\theta_{\mathrm{RX},k})+ \cos(\theta_{\mathrm{RX},k})\right)&
\\
1
\end{bmatrix}.
\label{eq:v_sk}
\end{equation}
is in the null space of $\left(\epsilon_{\mathbf{s}_k}+\beta_{\mathbf{s}_k}\left(1+\left\|\vec{p}-\vec{s}_k\right\|\right) \right) \vec{I}-\boldsymbol\Psi_{\mathbf{s}_k}$, i.e. 
\begin{equation}
\left(\left(\epsilon_{\mathbf{s}_k}+\beta_{\mathbf{s}_k}\left(1+\left\|\vec{p}-\vec{s}_k\right\|\right) \right) \vec{I}- \boldsymbol\Psi_{\mathbf{s}_k} \right)\mathbf{v}_{\mathbf{s}_k} = \vec{0}.
\label{eq:evec_null_space}
\end{equation}
Hence $\vec{v}_{\vec{s}_k}$ is the eigenvector corresponding to the eigenvalue in \eqref{eq:non_zero_eval_proof}.


%


%
%
%
%
%

\ifCLASSOPTIONcaptionsoff
  \newpage
\fi



%
\bibliographystyle{IEEEtran}
\bibliography{bibliography}
\end{document}